\documentclass[10pt,a4paper]{article}

\usepackage{kashsty}
\usepackage{kashThm}
\usepackage[round,authoryear]{natbib}

\usepackage[margin=1in]{geometry}

\usepackage{multirow}
\usepackage{algorithm}
\usepackage{algorithmic}

\def\PICDIR{./}

\title{Computation-free Nonparametric testing for
Local and Global Spatial Autocorrelation with application to 
the Canadian Electorate}
\author{
Adam B Kashlak\\
Weicong Yuan\\
\small Mathematical \& Statistical Sciences\\
\small University of Alberta\\
\small Edmonton, Canada,  T6G 2G1}

\begin{document}

\maketitle

\begin{abstract}
  Measures of local and global spatial association are
  key tools for exploratory spatial data analysis.
  Many such measures exist
  including Moran's $I$, Geary's $C$, and the Getis-Ord 
  $G$ and $G^*$ statistics.  A parametric approach to 
  testing for significance relies on strong assumptions,
  which are often not met by real world data.  Alternatively,
  the most popular nonparametric approach, the permutation test,
  imposes a large computational burden especially for massive 
  graphical networks.  Hence, we propose a computation-free
  approach to nonparametric permutation testing for local 
  and global measures of spatial autocorrelation 
  stemming from generalizations of the Khintchine inequality
  from functional analysis and the theory of $L^p$ spaces.
  Our methodology is demonstrated on the results of the
  2019 federal Canadian election in the province of Alberta.
  We recorded the percentage of the vote gained by the 
  conservative candidate in each riding.
  This data is not normal, and the sample size is fixed at 
  $n=34$ ridings making the parametric approach invalid.
  In contrast, running a classic permutation test for 
  every riding, for multiple test statistics, with various 
  neighbourhood structures, and multiple testing correction
  would require the simulation of millions of permutations.
  We are able to achieve similar statistical power on this 
  dataset to the permutation test without the need for 
  tedious simulation.
  We also consider data simulated across the entire 
  electoral map of Canada.
\end{abstract}

\tableofcontents

\section{Introduction}

The 2019 Canadian federal election left the province of Alberta
an homogeneous sea of blue as the Conservative Party swept the 
entire province except for the small riding of Edmonton-Strathcona
captured by Heather McPherson of the New Democratic Party.
Is the province merely an highly homogeneous  mass of 
conservatism or are there more features to the political 
topology?  In this article, we answer this question by
proposing a novel nonparametric
approach to testing for local and global spatial autocorrelation 
via an analytic variant of the classic permutation test.

Global and local indicators of spatial association 
are a cornerstone of exploratory spatial data analysis
\citep{ANSELIN1995,ANSELIN2019}.  For a connected graph 
$\mathcal{G}$ with $n$ vertices $\nu_1,\ldots,\nu_n$, 
set of edges $\mathcal{E}$, and random variables $y_1,\ldots,y_n$
associated with each vertex, a global indicator of 
spatial association (GISA) tests whether the random vector 
$\boldsymbol{y}=(y_1,\ldots,y_n)$ is uncorrelated or 
has some non-negligible spatial autocorrelation.
Similarly, a local indicator of spatial association (LISA)
tests whether or not random variable $y_i$ at vertex $\nu_i$
is correlated with some user-defined local neighbourhood of $\nu_i$.

Many measures of GISA and LISA have been proposed including 
Moran's $I$, Geary's $C$, and the Getis-Ord statistics;
see, for example,  
\cite{CLIFFORD1981,SOKAL1998,WALLER2004,GETISORDb,GAETAN2010,SEYAS2020} 
and others for more details.
Two standard testing paradigms exist for these statistics:
asymptotic normality and permutation tests.  The former
suffers from strong distributional assumptions.  
Furthermore, the assumption that $n\rightarrow\infty$
is not valid in this context; Alberta has a fixed 
$n=34$ ridings (vertices) without hope for increasing the
sample size.  In turn, the permutation test offers a powerful
nonparametric testing alternative to asymptotic normality.
The permutation test simply computes the value of the chosen
test statistic under uniformly random permutations of the 
observations.
Its downfall stems from the computation required, because
as we cannot enumerate the entire set of $n!$ elements of
the symmetric group, we instead
randomly draw permutations to get a Monte Carlo estimate of
the p-value.  This results in the dual problems of heavy
computation and estimated p-values that are upwardly biased
causing a loss in statistical power.  
For example, if we were to test for local 
autocorrelation at one vertex, we may want to simulate 10,000
permutations to get an accurate estimate of the p-value.  
Repeating this for Canada's 338 ridings, would result in 
3.38 million permutations.  Furthermore, applying the Bonferroni 
multiple testing correction would warrant another, say, 
100x permutations per vertex requiring 338 million in total.
Repeating this test with different neighbourhood designations---e.g.
$k=1,2,3$ nearest neighbours---would further multiply the 
computational burden.   

Our approach follows from the permutation test, but instead of 
proceeding via tedious Monte Carlo simulation, we instead 
use analytic formulations of the permutation test from 
the recent works of 
\cite{SPEKTOR2016,KASHLAK_KHINTCHINE2020,SUSANNAORLI2020}.
This allows us to propose an analytic formula for
computation of the permutation test p-value.
The underlying thread tying together the many measures
of local and global autocorrelation is that 
all of these statistics can be considered as a special 
case of the Gamma index for matrix association
\citep{MANTEL1967,HUBERT1985} as noted in 
\cite{ANSELIN1995}.  Hence, we derive a general p-value
bound applicable to any statistic falling into the
form of a Gamma index in Theorem~\ref{thm:localGamma}
and specify it to cases of interest for LISA and GISA testing
in Sections~\ref{sec:lisaTest} and~\ref{sec:gisaTest},
respectively.
Prior to that, the test statistics of interest are 
briefly introduced in Section~\ref{sec:methods}, and
subsequently the
Canadian electorate data is analyzed in 
Section~\ref{sec:dataAnalysis} along with simulated 
data.

\section{Methods for testing spatial association}
\label{sec:methods}

\subsection{Local indicators of spatial association}

For a graph $\mathcal{G}$ with $n$ vertices $\nu_1,\ldots,\nu_n$ 
and real valued measurements
$y_1,\ldots,y_n\in\real$ at each vertex, we can define a few different 
measurements of LISA
based on a user specified $n\times n$ weight matrix $W$.
See \cite{ANSELIN1995,BIVAND2018,SEYAS2020} for more details.

Local Moran's index for node $i$ is defined as 
$$
  I_i = \frac{y_i-\bar{y}}{\hat{\sigma}^2}\sum_{j=1}^n w_{i,j}(y_j-\bar{y})
$$
where $w_{i,j}$ is the $i,j$th entry in the chosen weight matrix
$W$ and $\hat{\sigma}^2 = n^{-1}\sum_{i=1}^n(y_i-\bar{y})^2$ is the
sample variance of the $y_i$.
In \cite{SOKAL1998}, moments for local Moran's index are derived
under the total randomization hypothesis
being ``the one under which all permutations of the observed data 
values on the locations are equally likely.''
These moments are 
$$
  \xv I_i = -\frac{w_{i,(1)}}{n-1}
  ~\text{ and }~
  \var{I_i} = w_{i,(2)}\frac{n-b}{n-1} +
  (w_{i,(1)}^2-w_{i,(2)})\frac{2b-n}{(n-1)(n-2)} -
  \frac{w_{i,(1)}^2}{(n-1)^2}.
$$
where $w_{i,(1)}=\sum_{j=1}^nw_{i,j}$ and
$w_{i,(2)}=\sum_{j=1}^nw_{i,j}^2$ and
$b = n\sum_{i=1}^ny_i^4 / (\sum_{i=1}^ny_i^2)^2$.

Local Geary's statistic for node $i$ is defined as
$$
  C_i = \frac{1}{\hat{\sigma}^{2}}\sum_{j=1}^n w_{i,j}(y_i-y_j)^2
$$
with $\hat{\sigma}^{2}$ as above.  It it preferable to perform
a significance test for $C_i$ using a permutation test as opposed to 
parametric methods \citep{ANSELIN1995,ANSELIN2019,SEYAS2020}.
Regardless, the first and second moments under the total randomization
hypothesis are
$$
  \xv C_i = \frac{2nw_{i,(1)}}{n-1} 
  ~\text{ and }~
  \var{C_i} = \left(\frac{n}{n-1}\right)(w_{i,(1)}^2-w_{i,(2)})(3+b) -
  \left(\frac{2nw_{i,(1)}}{n-1}\right)^2
$$
as outlined in \cite{SOKAL1998}.

The Getis--Ord statistics for node $i$ are
$$
  G_i   = \frac{\sum_{j\ne i} w_{i,j} y_j}{\sum_{j\ne i}y_j},
  ~~~
  G_i^* = \frac{\sum_{j} w_{i,j} y_j}{\sum_{j}y_j}.
$$
with means and variances
\begin{align*}
  \xv G_i &= \frac{w_{i,(1)}}{n-1}, \text{ and }&
  \var{G_i} &= \frac{w_{i,(1)}(n-1-w_{i,(1)})\hat{\sigma}^2_{-i}}{
  	(n-1)^2(n-2)\bar{y}^2_{-i}
  },\\
  \xv G_i^* &= \frac{w_{i,(1)}}{n}, \text{ and }&
  \var{G_i^*} &= \frac{w_{i,(1)}(n-w_{i,(1)})\hat{\sigma}^2}{
	n^2(n-1)\bar{y}^2
  }
\end{align*}
with $\bar{y}_{-i}$ and $\hat{\sigma}^2_{-i}$ being the sample
mean and variance of $y_1,\ldots,y_n$ with the $i$th data point
removed.
 
All of these statistics have been thoroughly discussed in the
noted references.  Briefly, Moran's I is the closest analogue
to autocorrelation from a time series context, which can be 
positive or negative depending on how the neighbourhood values
deviate above or below the sample mean.  It will be near
zero, however, if the local measurements lie close to the sample
mean whereas Geary's C will deem such a setting to have strong
positive spatial association.  The Getis--Ord statistics act
like moving averages identifying local clusters that all 
exhibit large values, which are sometimes referred to 
as ``hot spots'' in the literature.  
These two statistics, $G$ and $G^*$, behave similarly
to Moran's I.
 
\subsection{Global indicators of spatial association}
 
The above LISA statistics naturally extend to 
GISA statistics through summation.  Though chronologically,
GISA statistics came first, and LISA statistics were designed
in the following additive form:
\begin{align*}
  I = \sum_{i=1}^n I_i,~
  C = \sum_{i=1}^n C_i,~
  G = \sum_{i=1}^n G_i,~
  G^* = \sum_{i=1}^n G_i^*.
\end{align*}
Consequently, we both have a global measure and an 
ANOVA-like decomposition of the global spatial association
into individual contributions from each of the $n$ nodes
in the graph.  Similar to the LISA statistics, significance
of these GISA statistics can be established via 
their moments and a normal approximation or via 
a permutation test.  
 
\subsection{Permutation tests for spatial association}

To perform a permutation test for a LISA statistic at 
node $\nu_i$, we apply a restricted permutation of the 
nodes $\pi\in\mathbb{S}_n$ such that $\pi(i)=i$.  Thus,
we fix the $i$th node and permute the others.  
For $B\in\natural$ permutations, 
$\pi_1,\ldots,\pi_B$, drawn uniformly at random 
from $\mathbb{S}_n$, the symmetric group on $n$ elements, 
that fix the $i$th entry, the upper-tail p-value is estimated 
to be 
$$
  \text{p-value} = \frac{1}{B+1}\left(
    1+\sum_{k = 1}^B \boldsymbol{1}\left[
      T_i(\pi_k) \ge T_i^*
    \right] 
  \right)
$$
where $T_i^*$ is the chosen statistic of interest---e.g. 
$I_i$ or $C_i$---and $T_i(\pi_k)$ is the value of that 
statistic computed after permuting the measurements 
$y_1,\ldots,y_n$ by $\pi_k$.  Similarly, the lower tail
can be computed by reversing the inequality.

We note that while Moran's $I$ and the Getis--Ord $G$ and $G^*$ 
are distinct statistics, in the context of a permutation 
test, they all yield the same inference.  This is because
the ordering of the $T_i(\pi_k)$ is preserved whether 
$y_j$ or $y_j-\bar{y}$ appears in the summand.

A permutation test for GISA statistics is applied similarly
to those for LISA statistics. 
One area of dispute is whether the randomization should
be restricted or unrestricted sometimes referred to as 
conditional or total randomization, respectively.
Mathematically, unrestricted randomization would 
consider uniformly random permutations in $\mathbb{S}_n$
while restricted randomization would consider uniformly
random permutations from $\mathbb{S}_n$ that fix $i$.
This is discussed in more detail in 
\cite{ANSELIN1995} and \cite{SOKAL1998} among others.
 In the theoretical development of Section~\ref{sec:theory},
 we consider the restricted randomization setting.
 
\section{Theory}
\label{sec:theory}

\subsection{Analytic permutation test for the Local Gamma Index}

The gamma index \citep{MANTEL1967,HUBERT1985} is a general measure of 
matrix association defined for two similar, say $n\times n$, 
matrices $A$ and $B$ with entries
$a_{i,j}$ and $b_{i,j}$, respectively, to be 
$
  \gamma_{AB} := \sum_{i,j=1}^na_{i,j}b_{i,j}
$
where we use the notation $\gamma$ instead of the more standard $\Gamma$ to avoid conflict
with the use of the gamma function below.  
Typically, the $a_{i,j}$ and $b_{i,j}$ can be treated as measures of proximity
between objects $i$ and $j$ resulting in $\gamma_{AB}$ being an unnormalized
measure of association (or correlation) between matrices $A$ and $B$.  
In \cite{HUBERT1985}, it is shown how the gamma index can be seen as a general
correlation measure, which includes many classic correlation statistics such 
as Pearson correlation, Spearman's $\rho$, and Kendall's $\tau$.  As a 
result, the gamma index is sometimes referred to as the general correlation 
coefficient.

The local gamma index introduced in \cite{ANSELIN1995} is a local version
of the above gamma index defined as 
$
  \gamma_i = \sum_{j=1}^n a_{i,j}b_{i,j}
$
where $A$ and $B$ are dropped for notational convenience.  We note that 
$\gamma = \sum_{i=1}^n \gamma_{i}$ thus decomposing the global gamma index
into a sum of local gamma indices reminiscent of ANOVA.  
For specific choices of 
$a_{i,j}$ and $b_{i,j}$, \cite{ANSELIN1995} shows that the  local gamma index 
can be specified to local Moran's, Geary's, and the Getis-Ord statistics
as well as others.  This is achieved by noting that each of these
statistics can be written as the gamma index between a weight
matrix $W$---e.g. the 
adjacency matrix---and a data association matrix $\Lambda$---e.g.
$\lmb_{i,j} = y_iy_j$.
Thus, we focus our theoretical development on the local gamma index.

In Theorem~\ref{thm:localGamma}, 
we develop analytic bounds on the permutation test statistic's 
p-value via 
application of a weakly dependent variant of the Khintchine inequality
\citep{HAAGERUP1981,GARLING2007,SPEKTOR2016,KASHLAK_KHINTCHINE2020,SUSANNAORLI2020}.
In what follows, $W$ is a binary weight matrix---i.e. $w_{i,j}\in\{0,1\}$---with 
diagonal entries of zero.  Such $W$ include the adjacency matrix for the 
graph $\mathcal{G}$ as well as the $k$-nearest-neighbours matrix where 
$w_{i,j}=1$ if there exists a path from $\nu_i$ to $\nu_j$ of length
no greater than $k$.
In Theorem~\ref{thm:localGamma}, 
we require the below low-connectivity condition on the weight matrix,
which is 
reasonable for large planar graphs as considered in the data from 
Section~\ref{sec:dataAnalysis}.   However, this condition 
can also be reversed as is discussed in Remark~\ref{rmk:highConnect}.
\begin{condition}[No highly connected vertices]
	\label{cond:noConVert}
	For each row $i$ of W, $\sum_{j=1}^nw_{i,j} \le n/2$.
\end{condition} 

\begin{theorem}[Local Gamma Index]
  \label{thm:localGamma}
  For a graph $\mathcal{G}$ with $n$ vertices, let $W$ be a binary-valued $n\times n$
  weight matrix with zero diagonal, and let $\Lambda$ be an $n\times n$ matrix 
  with entries $\lmb_{i,j} = \lmb( y_i, y_j )$ with $\lmb:\real^2\rightarrow\real$ 
  being a measure of proximity---e.g. $\lmb( y_i, y_j )=(y_i-\bar{y})(y_j-\bar{y})$ 
  for Moran or $\lmb( y_i, y_j )=(y_i-y_j)^2$ for Geary.  The local gamma index between 
  $W$ and $\Lambda$ at vertex $i$ is $\gamma_{i} = \sum_{j=1}^n w_{i,j}\lmb( y_i, y_j )$
  and the permuted variant of this test statistic is 
  $
    \gamma_{i}(\pi) = \sum_{j=1}^n w_{i,j}\lmb( y_i, y_{\pi(j)} )
  $
  where $\pi$ is a uniformly random element of $\mathbb{S}_{n}$ conditioned 
  so that $\pi(i)=i$.  Then, for vertex $i$ under Condition~\ref{cond:noConVert} 
  denoting $m_i = \sum_{j=1}^n w_{i,j}$, 
  $\bar{\lmb}_{-i} = (n-1)^{-1}\sum_{j\ne i}\lmb_{i,j}$,  and 
  $s_i^2=(n-1)^{-1}\sum_{j\ne i}(\lmb_{i,j}-\bar{\lmb}_{-i})^2$,
  \begin{equation}
  \label{eqn:LGBound}
    \prob{
      \abs{\gamma_i{(\pi)}-m_i\bar{\lmb}_{-i}} \ge \gamma_i \,|\,
      y_1,\ldots,y_n
    } \le 
    \exp\left( 
    -\frac{m_i\gamma_i^2}{2s_i^2(n-m_i-1)^2}
    \right)
  \end{equation}
  Furthermore, 
  \begin{equation}
  \label{eqn:LGBetaBound}
    \prob{ \abs{\gamma_i{(\pi)}-m_i\bar{\lmb}_{-i}} \ge \gamma_i \,|\,
    	y_1,\ldots,y_n } \le 
    C_0 I\left[
      \exp\left( 
      -\frac{m_i\gamma_i^2}{2s_i^2(n-m_i-1)^2}
      \right);
      \frac{(n-1)(n-m_i-1)}{m_i^2},\frac{1}{2}
    \right]
  \end{equation}
  where $I[\cdot]$ is the regularized incomplete beta function and
  $$
  C_0 = \frac{
  	{\sqrt{(n-1)(n-m_i-1)}}\Gamma\left(\frac{(n-1)(n-m_i-1)}{m_i^2}\right)
  }{
  	m_i\Gamma\left(\frac{1}{2}+\frac{(n-1)(n-m_i-1)}{m_i^2}\right)
  }
  $$
  with $\Gamma(\cdot)$ the gamma function.  
\end{theorem}

\begin{remark}[Highly connected vertex]
  \label{rmk:highConnect}
  In the proof of Theorem~\ref{thm:localGamma}, the assumption that 
  $m_i \le n-m_i-1$ from Condition~\ref{cond:noConVert} is used.  If
  the converse were true, then the proof can be rerun by swapping
  the roles of $m_i$ and $n-m_i-1$.  The resulting bounds
  are
  $$
    \prob{
    	\abs{\gamma_i{(\pi)}-m_i\bar{\lmb}_{-i}} \ge \gamma_i \,|\,
    	y_1,\ldots,y_n
    } \le 
    \exp\left( 
      -\frac{(n-m_i-1)\gamma_i^2}{2s_i^2m_i^2}
    \right)
  $$
  and 
  $$
  \prob{ \abs{\gamma_i{(\pi)}-m_i\bar{\lmb}_{-i}} \ge \gamma_i \,|\,
  	y_1,\ldots,y_n } \le
  C_0 I\left[
  \exp\left( 
    -\frac{(n-m_i-1)\gamma_i^2}{2s_i^2m_i^2}
  \right);
  \frac{(n-1)m_i}{(n-m_i-1)^2},\frac{1}{2}
  \right]
  $$
  with
  $$
  C_0 = \frac{
  	{\sqrt{(n-1)m_i}}\Gamma\left(\frac{(n-1)m_i}{(n-m_i-1)^2}\right)
  }{
  	(n-m_i-1)\Gamma\left(\frac{1}{2}+\frac{(n-1)m_i}{(n-m_i-1)^2}\right)
  }.
  $$
\end{remark}

The necessity of having a binary-valued weight matrix arises
from the proof of Theorem~\ref{thm:localGamma}, which reframes
testing for significant spatial association as a two sample 
test---i.e. the 0's and the 1's delineate two samples to 
compare.  Though, unlike the 
two sample tests discussed in \cite{KASHLAK_KHINTCHINE2020},
we typically have many more 0-weights than 1-weights as 
indicated in Condition~\ref{cond:noConVert}.  Hence,
the first bound in Equation~\ref{eqn:LGBound} is mathematically
valid but excessively
conservative in practice.  The beta-corrected bound in 
Equation~\ref{eqn:LGBetaBound} rectifies this problem
as demonstrated with both real and simulated data in 
Section~\ref{sec:dataAnalysis}.  As an alternative
to this beta-correction, \cite{KASHLAK_KHINTCHINE2020}
also proposes an empirical correction based on performing
a small number of permutations to estimate the 
parameters for the incomplete beta function.  We 
further note in the simulations in Section~\ref{sec:dataAnalysis}
that this approach is inferior to our formula
presented in Equation~\ref{eqn:LGBetaBound}.

\subsection{Extensions to testing LISA}
\label{sec:lisaTest}

By directly applying Theorem~\ref{thm:localGamma} to Moran's $I$
and Geary's $C$, we have the below corollaries. 
In practice, one should use Equation~\ref{eqn:LGBetaBound}
for significance testing.  Nevertheless, the following 
sub-Gaussian bounds give intuition regarding the behaviour 
of these LISA statistics.

\begin{corollary}[Moran's Statistic]
  \label{cor:moranLisa}
  For $I_i = \frac{y_i-\bar{y}}{\hat{\sigma}^2}\sum_{j=1}^nw_{i,j}(y_j-\bar{y})$,
  the permutation test p-value is bounded by
  $$
    \prob{
    	\abs{I_i{(\pi)}-m_i\bar{I}_{-i}} \ge I_i \,|\,
    	y_1,\ldots,y_n
    } \le 
    \exp\left( 
      -\frac{m_iI_i^2}{2(n-m_i-1)^2}
      \left(\frac{\hat{\sigma}^4}{s_i^2}\right)
    \right)
  $$
\end{corollary}

\begin{corollary}[Geary's Statistic]
  \label{cor:gearyLisa}
  For $C_i = \frac{1}{\hat{\sigma}^2}\sum_{j=1}^nw_{i,j}(y_i-y_j)^2$,
  the permutation test p-value is bounded by
  $$
    \prob{
      \abs{C_i{(\pi)}-m_i\bar{C}_{-i}} \ge C_i \,|\,
      y_1,\ldots,y_n
    } \le 
    \exp\left( 
      -\frac{m_iC_i^2}{2(n-m_i-1)^2}
      \left(\frac{\hat{\sigma}^4}{s_i^2}\right)
    \right)
  $$
\end{corollary}

We note that the only difference between these corollaries and
Theorem~\ref{thm:localGamma} is the inclusion of a factor of 
$\hat{\sigma}^4$ in the numerator.  In the case of Moran's I
after centring so that $\bar{\lmb}_{-i}=0$,
the ratio with $s_i^2$ becomes
$$
  \frac{\hat{\sigma}^4}{s_i^2} =
  (n-1)^{-1}\frac{
  	\sum_{i,j=1}^n (y_i-\bar{y})^2(y_j-\bar{y})^2
  }{
    (y_i-\bar{y})^2\sum_{j=1,j\ne i}^n (y_j-\bar{y})^2
  }.
$$
As the double sum in the numerator is over $n^2$ terms, the sum in 
the denominator can be thought of summing along the $i$th row
of these values.  Hence, the p-value becomes smaller if the 
$i$th row of entries has lower variance than the average row 
variance and vice versa.  In the case of Geary's C after 
centring about $\bar{\lmb}_{-i}$, we have  
$s_i^2 = (n-1)^{-1}\sum_{j=1,j\ne i}^n(y_i-y_j)^4$, 
which yields a similar intuition.

\begin{remark}[Getis-Ord Statistics]
  As discussed in Section~\ref{sec:methods} and in 
  more detail in \cite{ANSELIN1995},
  a conditional/restricted permutation test on 
  local Moran's index will give an identical
  empirical reference distribution to a permutation test on either 
  of the Getis-Ord $G$ or $G^*$ statistics.  This is because for
  vertex $\nu_i$, we permute with $\pi\in\mathbb{S}_n$ such that 
  $\pi(i)=i$---that is, the permutation is a conditional randomization
  that fixes $y_i$---and this permutation only modifies the term
  $\sum_{j=1}^nw_{i,j}y_{\pi(j)}$.
\end{remark}

\subsection{Extensions to testing GISA}
\label{sec:gisaTest}

Global indicators can be defined as scaled sums of
corresponding local indicators as noted in Section~\ref{sec:methods} 
and in \cite{ANSELIN1995}.  To extend Theorem~\ref{thm:localGamma}
for LISA to Theorem~\ref{thm:globalGamma} for GISA, we must 
first define $\mathbb{S}_n^{\otimes n}$ to be the direct product
of $n$ copies of the symmetric group $\mathbb{S}_n$, which itself
satisfies the axioms of a group.

\begin{theorem}[Global Gamma Index]
	\label{thm:globalGamma}
	For a graph $\mathcal{G}$ with $n$ vertices, let $W$ be a binary-valued $n\times n$
	weight matrix with zero diagonal, and let $\Lambda$ be an $n\times n$ matrix 
	with entries $\lmb_{i,j} = \lmb( y_i, y_j )$ with $\lmb:\real^2\rightarrow\real$ 
	being a measure of proximity---e.g. $\lmb( y_i, y_j )=(y_i-\bar{y})(y_j-\bar{y})$ 
	for Moran or $\lmb( y_i, y_j )=(y_i-y_j)^2$ for Geary.  The gamma index between 
	$W$ and $\Lambda$ is 
	$\gamma = \sum_{i=1}^n \gamma_i =  
	\sum_{i=1}^n\sum_{j=1}^n w_{i,j}\lmb( y_i, y_j )$
	and the permuted variant of this test statistic is 
	$
	\gamma(\boldsymbol{\pi}) = \sum_{i=1}^n \gamma_i(\pi_i)
	$
	where $\boldsymbol{\pi}=(\pi_1,\ldots,\pi_n)$ 
	is a uniformly random element of $\mathbb{S}_{n}^{\otimes n}$ conditioned 
	so that $\pi_i(i)=i$.  Then, 
	denoting $m_i = \sum_{j=1}^n w_{i,j}$, 
	$\bar{\lmb}_{-i} = (n-1)^{-1}\sum_{j\ne i}\lmb_{i,j}$,  
	$\eta_i = m_i(n-m_i-1)/n-1$,
	and	$\upsilon^2=\sum_{i=1}^n \eta_i s_i^2$,
	$$
	\prob{ 
	  \abs*{\gamma(\boldsymbol{\pi})- \sum_{i=1}^n m_i\bar{\lmb}_{-i} } \ge \gamma 
	  \,\mid\, y_1,\ldots,y_n
    }
	\le
	\frac{1}{\sqrt{\pi}}\Gamma\left(
	\frac{\gamma^2}{4\upsilon^2}
	;\frac{1}{2}
	\right) + O(2^{-2n})
	$$
	with $\Gamma(\cdot;\cdot)$ the upper incomplete gamma function.  
\end{theorem}

\section{Data Analysis}
\label{sec:dataAnalysis}

\subsection{Simulated Data}

\subsubsection{Local Statistics}

Before delving into the results of the 2019 Canadian election,
we use the map of Canada's 338 ridings to simulate independent
data to verify correct performance of our methodology under the
null hypothesis of no spatial autocorrelation.  Hence, we 
simulate 338-long vectors of iid random variates coming from
both the standard normal distribution and the exponential 
distribution
with rate parameter set to 1.  For each of the 30 replications, 
we produce 338 p-values using Theorem~\ref{thm:localGamma}
with the incomplete beta function transformation for both 
Moran's and Geary's statistic.  The weight matrix $W$
is chosen to be the graph adjacency matrix.

The results of these four simulations are displayed in 
Figure~\ref{fig:canSim} in the form of QQ-plots 
charting the 338 ordered p-values against the expected
quantiles.  In all cases, these empirical values do not
deviate significantly from the main diagonal.  
In Table~\ref{tab:simCan}, we compare the performance of
Theorem~\ref{thm:localGamma} using the incomplete beta 
transform to three other methods for producing p-values: Theorem~\ref{thm:localGamma} using the empirical adjustment
mentioned in Section~\ref{sec:theory};
the standard computation-based permutation test using 1000 
permutations at each riding; and approximation of the test
statistic using the normal distribution based on the means 
and variances detailed in \cite{SOKAL1998}.
This is done by using the Anderson--Darling goodness of fit
test to test for uniformity of the 338 null p-values in 
each of the 30 replications.  Table~\ref{tab:simCan}
tabulates how many of these 30 replications are rejected 
as not uniform by the Anderson--Darling test at the 1\%
level.  
We see that the p-values produced by our methodology 
more often appear uniform than either the computation-based
permutation test or approaching our methodology via an empirical
adjustment.  Lastly, all 30 simulations producing p-values based
on Z-scores are rejected indicating that the normal approximation
for the distribution of either Moran's or Geary's statistic 
is not valid in this setting.

\begin{table}
  \begin{center}
  	\begin{tabular}{llrrrr}
  		\hline
  		&& \multicolumn{4}{c}{\bf Number of Anderson-Darling Rejections}\\
  		&& Beta Adjusted & Emp Adjusted & Computed Perms & Z Score \\
  		\hline
  		\multirow{ 2}{*}{Moran} &
  		Gaussian    & 2 & 14 & 8 & 30 \\
  		&Exponential & 3 & 8  & 16& 30 \\
  		\multirow{ 2}{*}{Geary} &
  		Gaussian &    0 & 12 &  5& 30 \\
  		&Exponential & 6 &  8 & 12& 30\\
  		\hline
  	\end{tabular}
  \end{center}
  \caption{
  	\label{tab:simCan}
  	The number of simulated null data sets whose set of 338
  	p-values are rejected by the Anderson--Darling goodness of 
  	fit test for uniformity at the 1\% level.  
  	30 replicates were performed for 
  	each case.
  }
\end{table}

\begin{figure}
  \begin{center}
    \includegraphics[width=0.45\textwidth]{\PICDIR/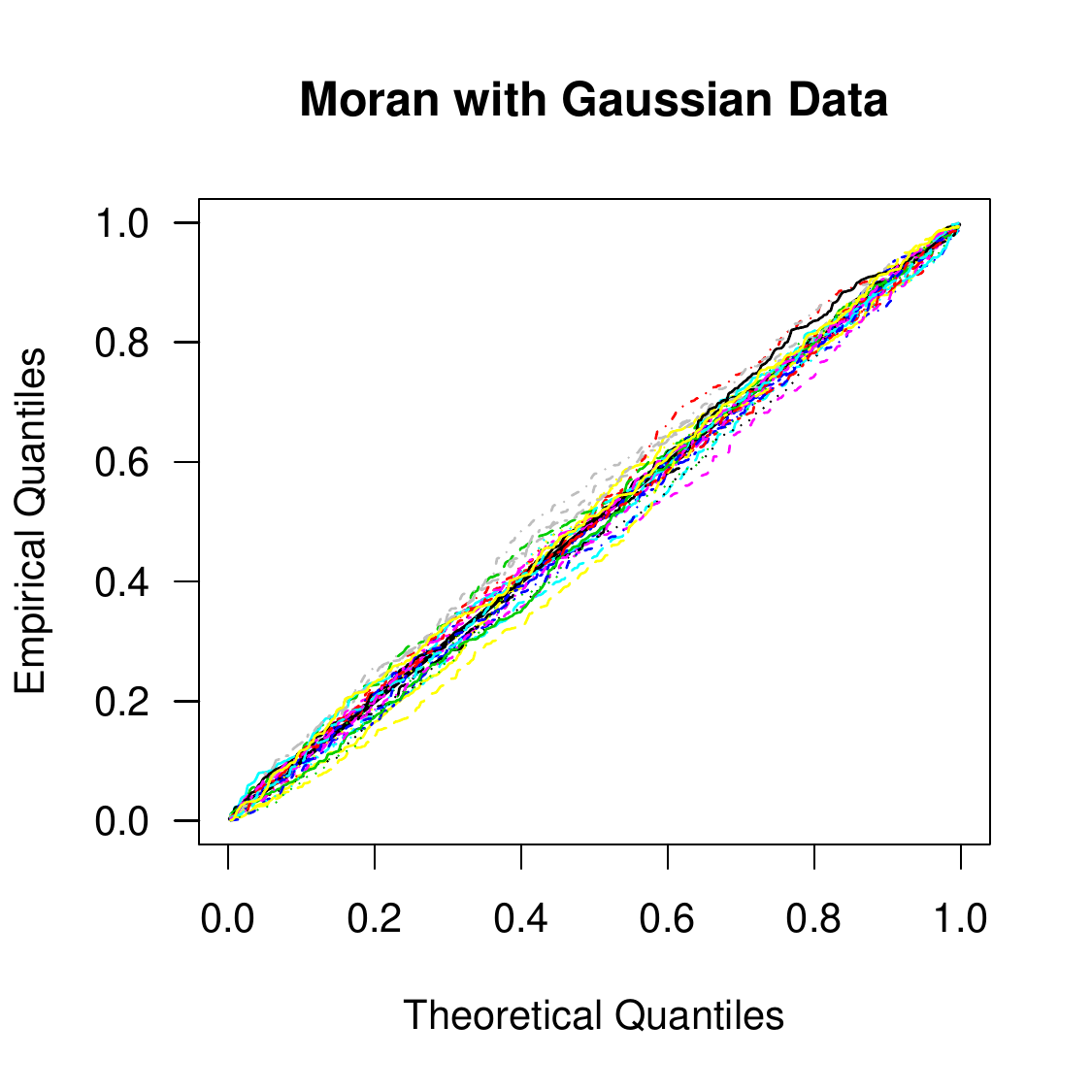}   
    \includegraphics[width=0.45\textwidth]{\PICDIR/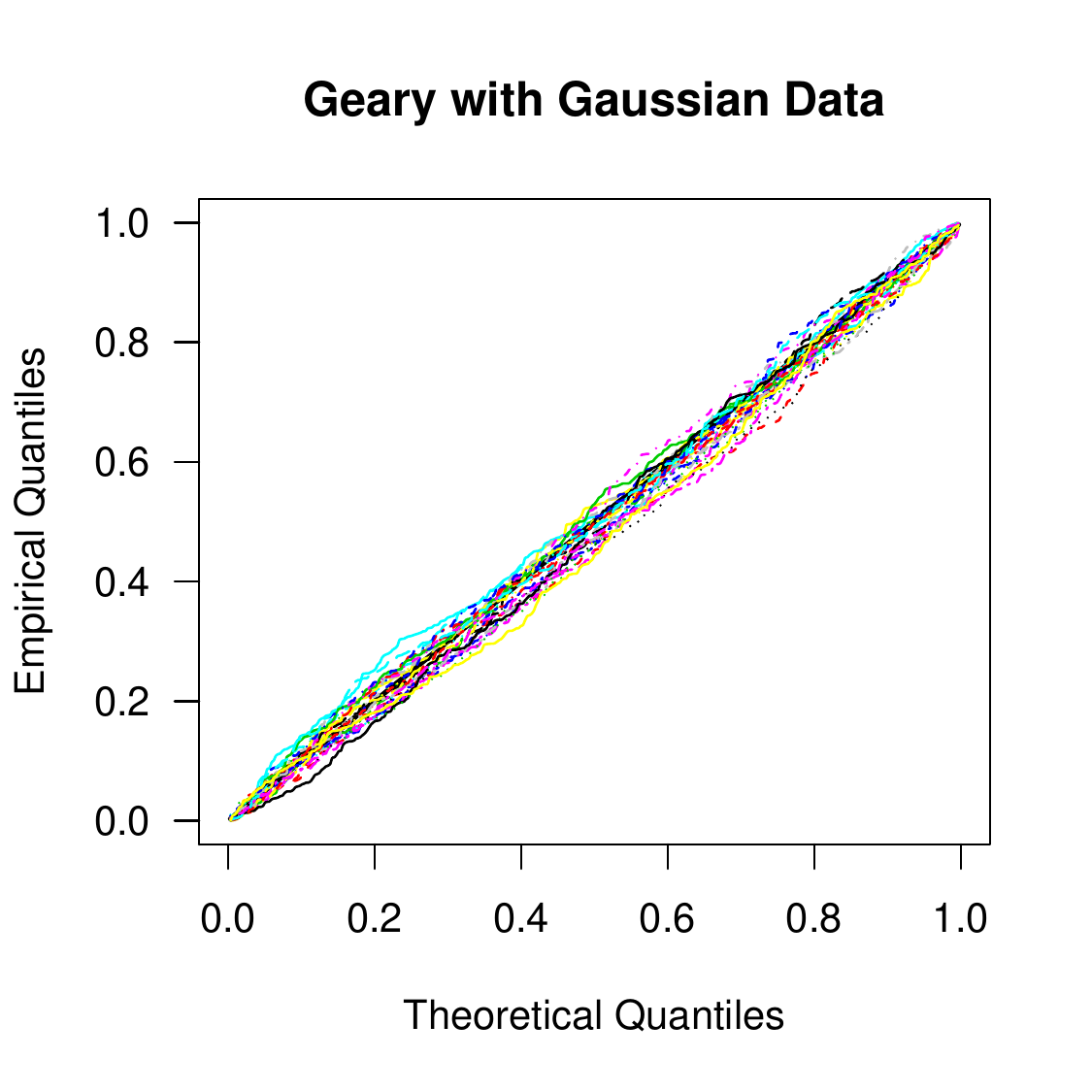}
    \includegraphics[width=0.45\textwidth]{\PICDIR/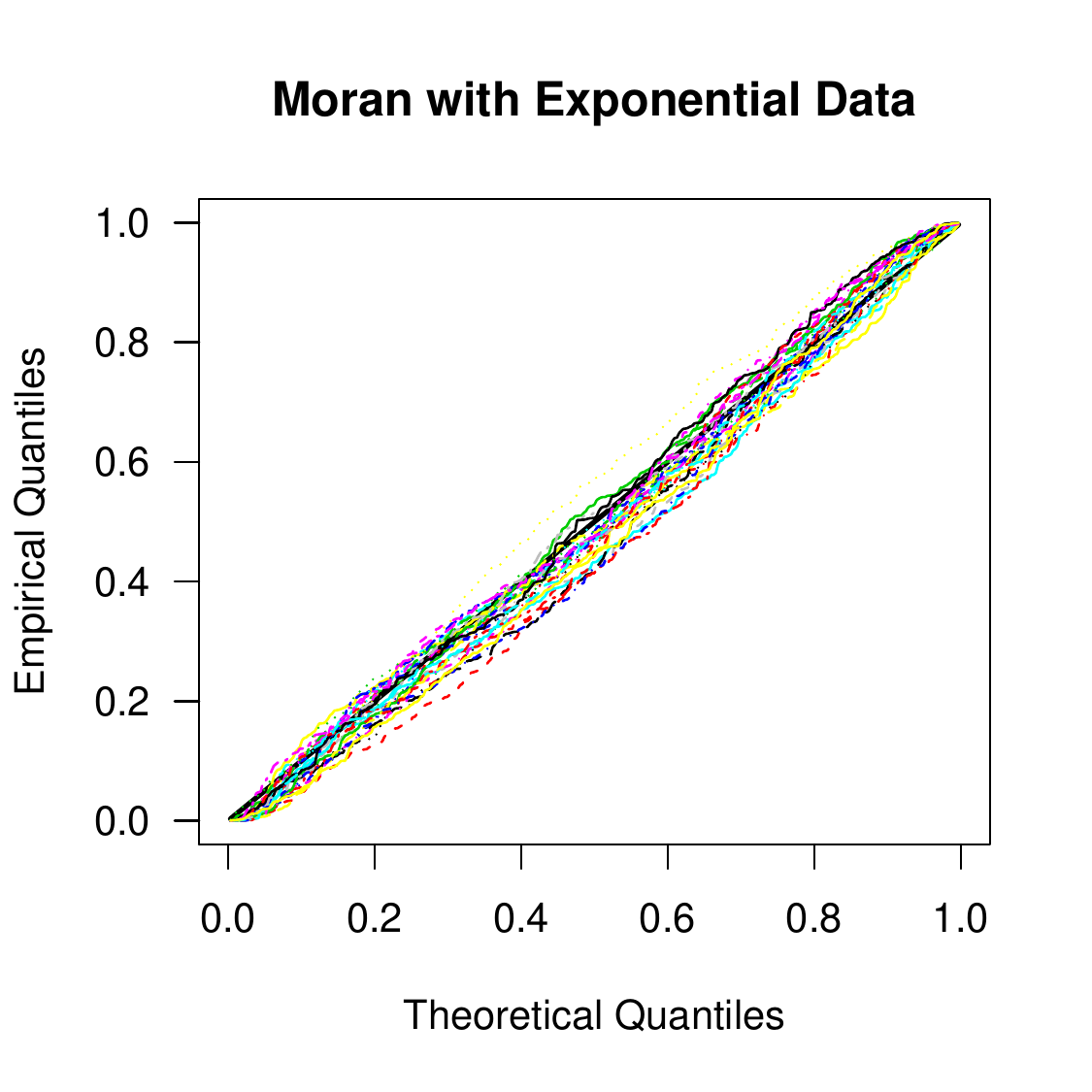}
    \includegraphics[width=0.45\textwidth]{\PICDIR/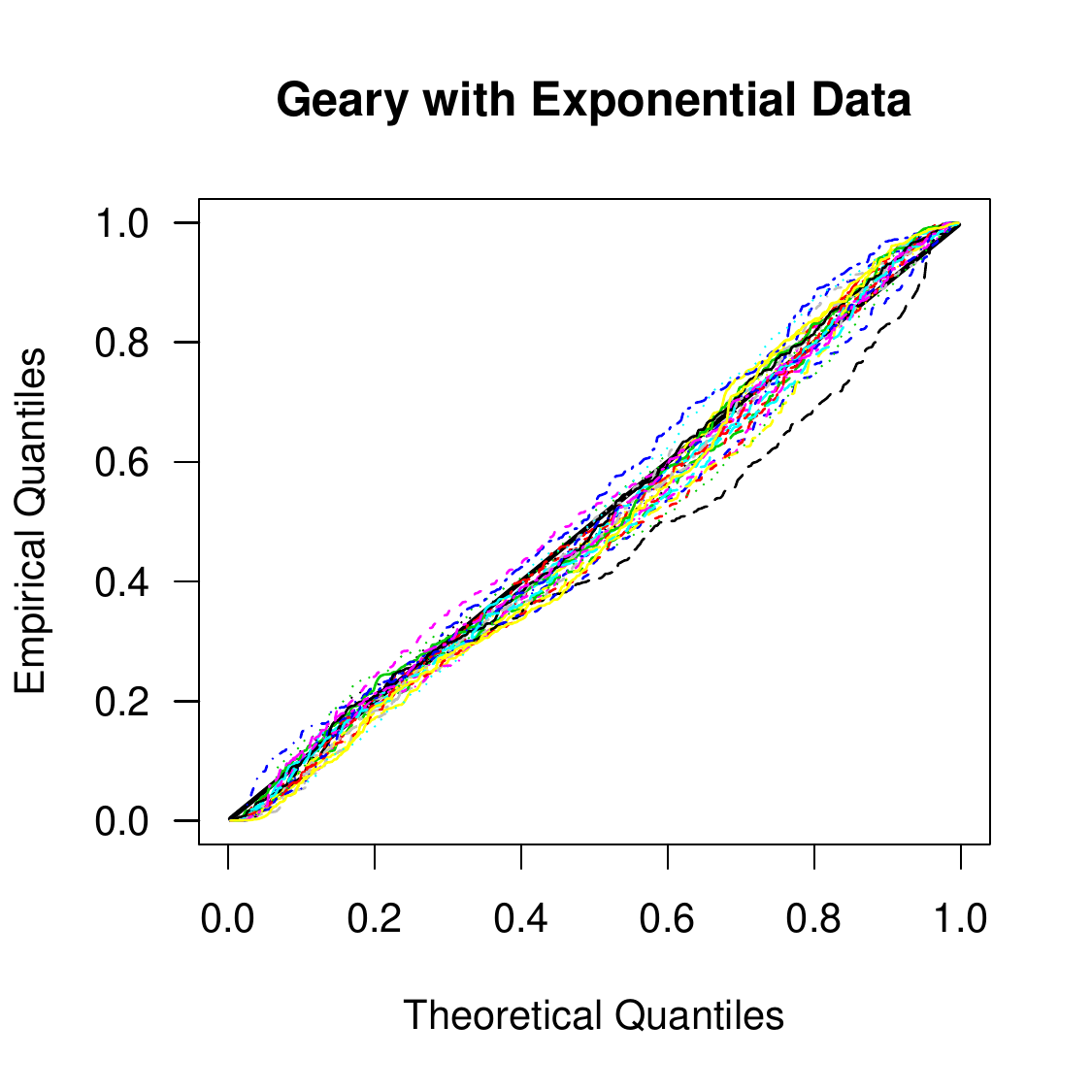}
  \end{center}
  \caption{
  	\label{fig:canSim}
    These plots depict 30 replicates of ordered p-values for Moran's 
  	and Geary's statistic and for independent Gaussian and 
  	exponential data simulated on the entire map of Canada's 
  	338 ridings.  This indicates that 
  	Theorem~\ref{thm:localGamma} produces p-values as would
  	be expected in the null setting of independence.
  }
\end{figure}

\subsubsection{Global Statistics}

We also test the performance of Theorem~\ref{thm:globalGamma}
in the null setting by simulating iid Gaussian and exponential 
data on the entire map of Canada.  Figure~\ref{fig:glbSim}
displays the results of 400 replications of each of the 
four settings.  For Moran's statistic in both cases and
Geary's statistic for Gaussian data, the distribution of
the 400 p-values is uniform as desired both visually 
and via the Kolmogorov-Smirnov and Anderson-Darling tests.
In the case of Geary with exponential data, the p-values 
produced by Theorem~\ref{thm:globalGamma} under-report
the significance---i.e. the p-values are larger than 
they should be.  This is corrected via the following 
empirical beta transform detailed in 
Algorithm~\ref{algo:empBetaAlg}, which is
similar to the one proposed in \cite{KASHLAK_KHINTCHINE2020},
but modified to handle GISA statistics.

\begin{algorithm}
	\caption{
		\label{algo:empBetaAlg}
		The Empirical Beta Transform for GISA Statistics
	}
	\begin{tabbing}
		\qquad \enspace Compute p-value 
		$p_0 = \frac{1}{\sqrt{\pi}}\Gamma(
		\frac{\gamma^2}{4\upsilon^2}
		;\frac{1}{2}
		)$ based on $\gamma$ chosen from the 
		desired GISA statistic.\\
		\qquad \enspace Choose $r>1$, the number of permutations 
		to simulate---e.g. $r=10$.\\
		\qquad \enspace Draw $\boldsymbol{\pi}_1,\ldots,
		\boldsymbol{\pi}_r$ from 
		$\mathbb{S}_n^{\otimes n}$
		uniformly at random under the condition that 
		$\boldsymbol{\pi}_{i,j}(j)=j$.\\
		\qquad \enspace  Compute $r$ p-values by
		$p_i=\frac{1}{\sqrt{\pi}}\Gamma(
		\frac{\gamma(\boldsymbol{\pi}_i)^2}{4\upsilon^2}
		;\frac{1}{2}
		)$ .\\ 
		\qquad \enspace Find the method of moments estimator
		for $\alpha$ and $\beta$ from the beta distribution.  \\
		\qquad\qquad Estimate first and second central moments of the 
		$p_i$ by $\bar{p}$ and $s^2$,\\ 
		\qquad\qquad\enspace the sample mean and variance.\\ 
		\qquad\qquad Estimate 
		$\hat{\alpha} = {\bar{p}^2(1-\bar{p})}/{s^2}-\bar{p}$.\\
		\qquad\qquad Estimate
		$\hat{\beta} = [ \bar{p}(1-\bar{p})/s^2 - 1 ][ 1-\bar{p} ]$.\\
		\qquad \enspace Return the adjusted p-value
		$I(p_0; \hat{\alpha},\hat{\beta})$.
		
	\end{tabbing}
\end{algorithm}

\begin{figure}
	\begin{center}
		\includegraphics[width=0.45\textwidth]{\PICDIR/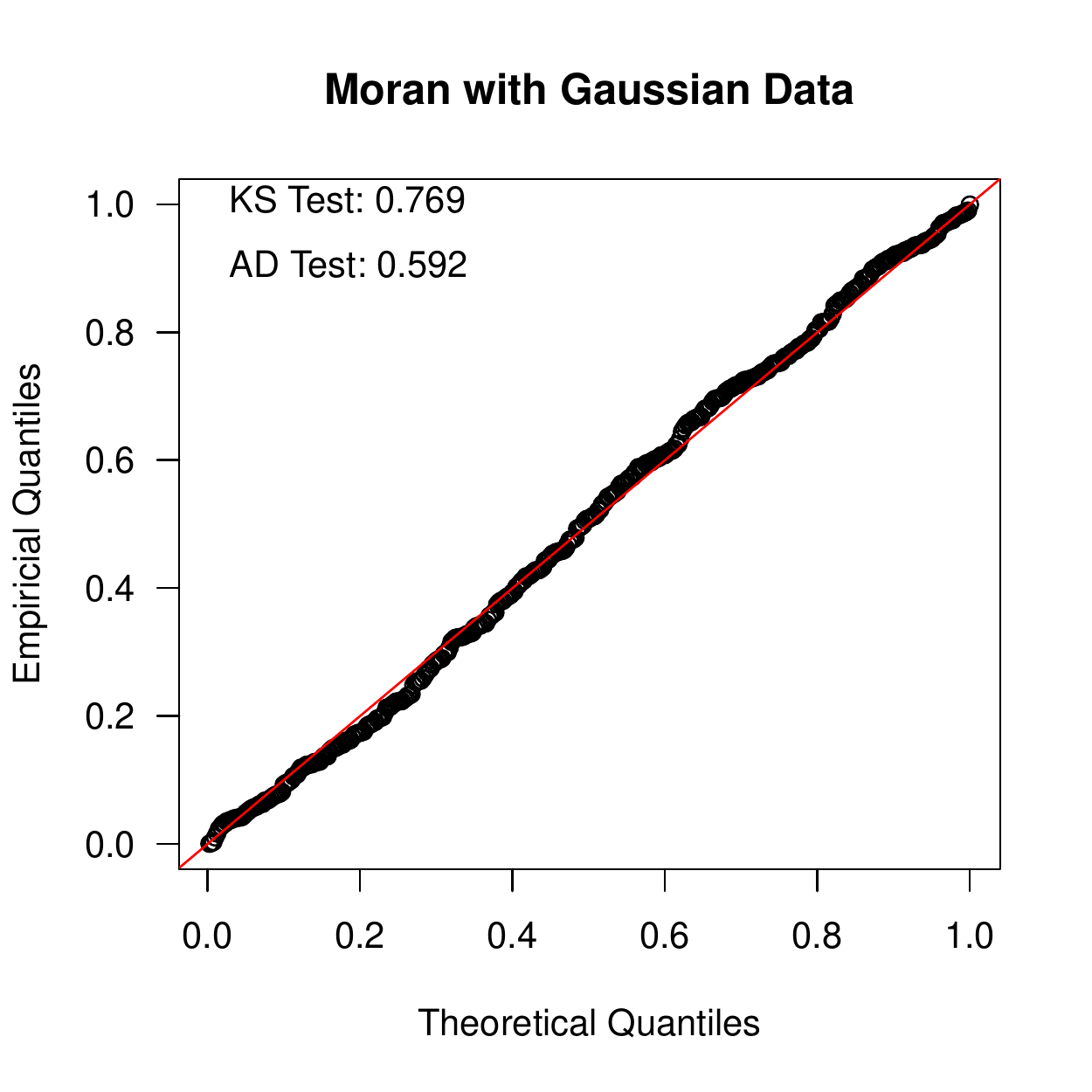} 
		\includegraphics[width=0.45\textwidth]{\PICDIR/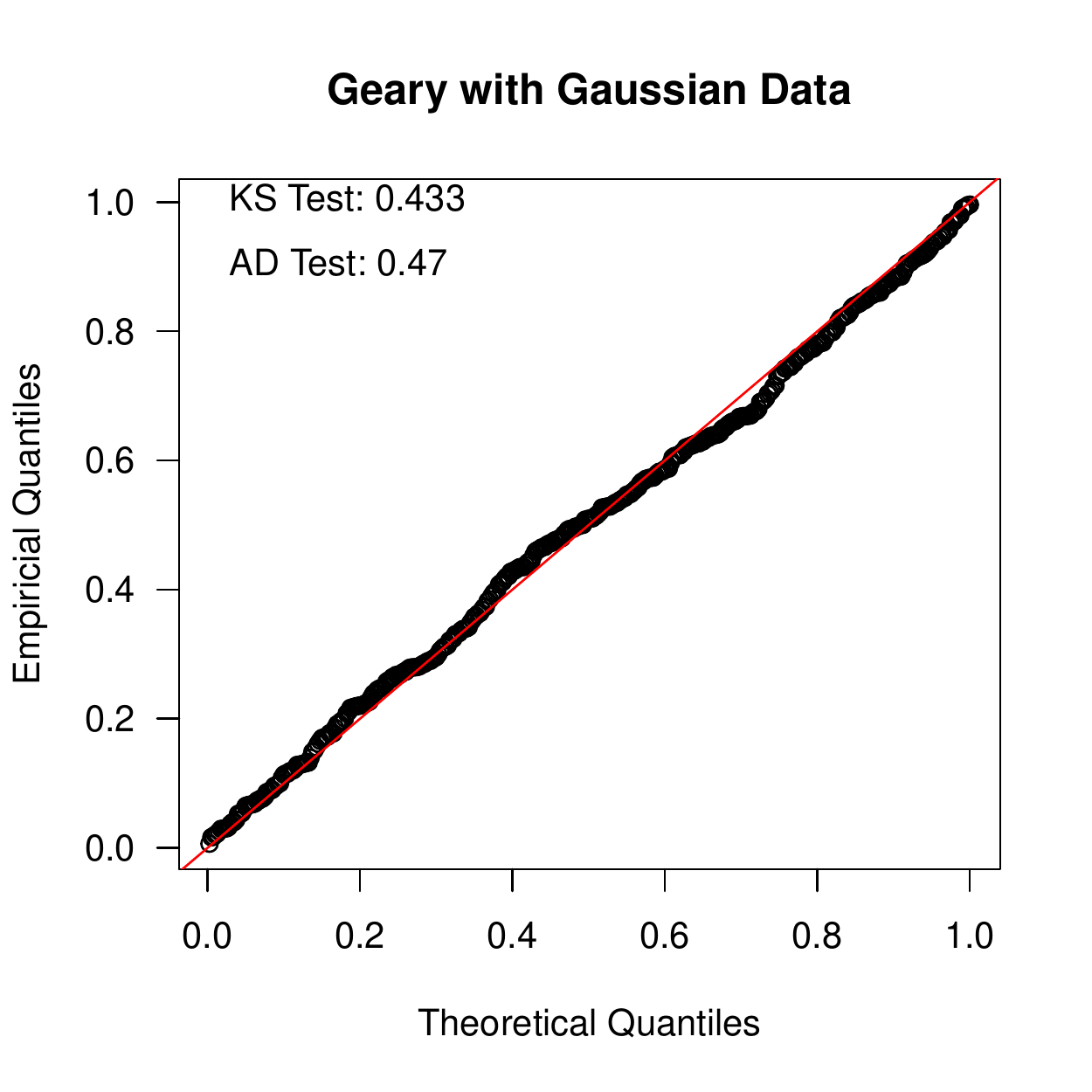}
		\includegraphics[width=0.45\textwidth]{\PICDIR/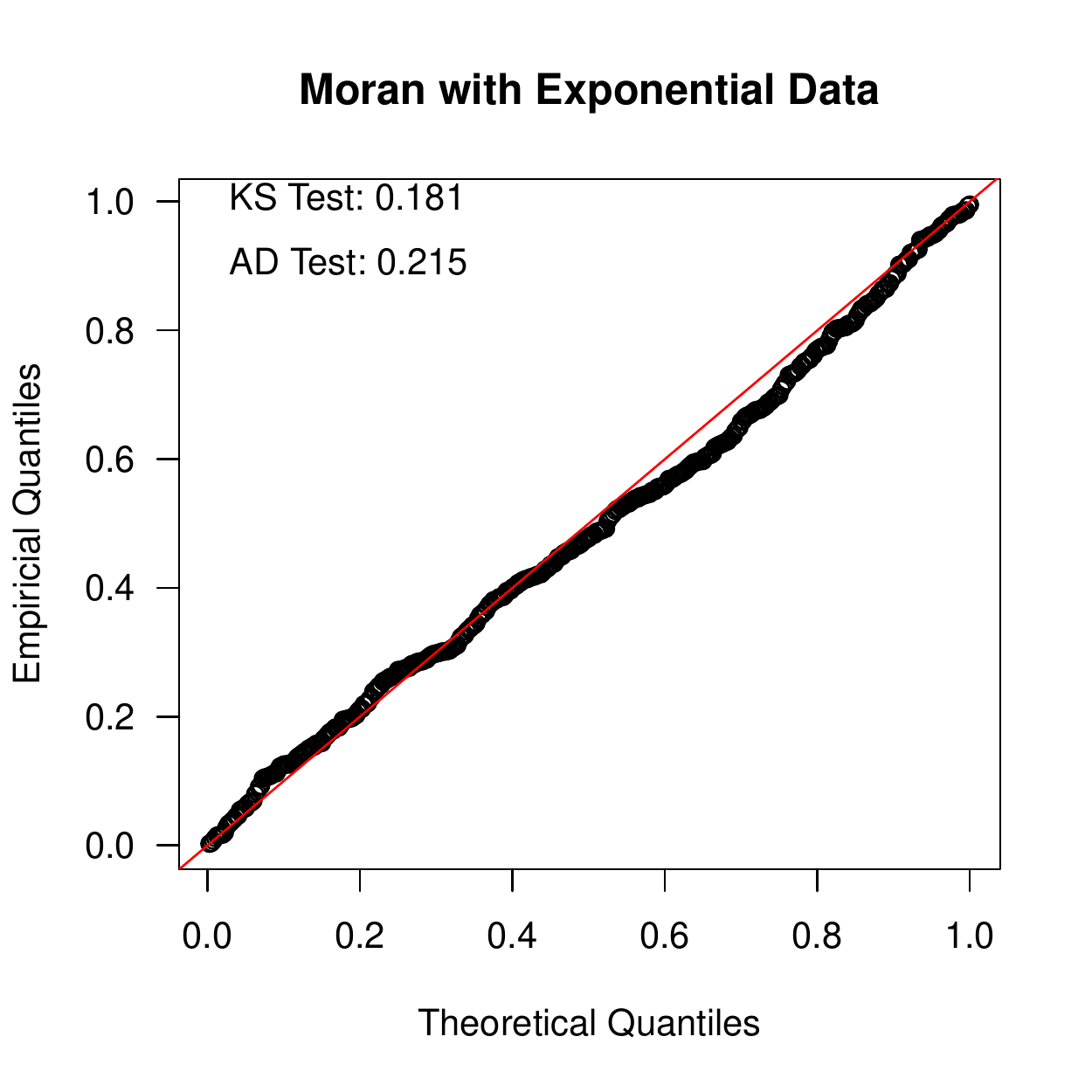}
		\includegraphics[width=0.45\textwidth]{\PICDIR/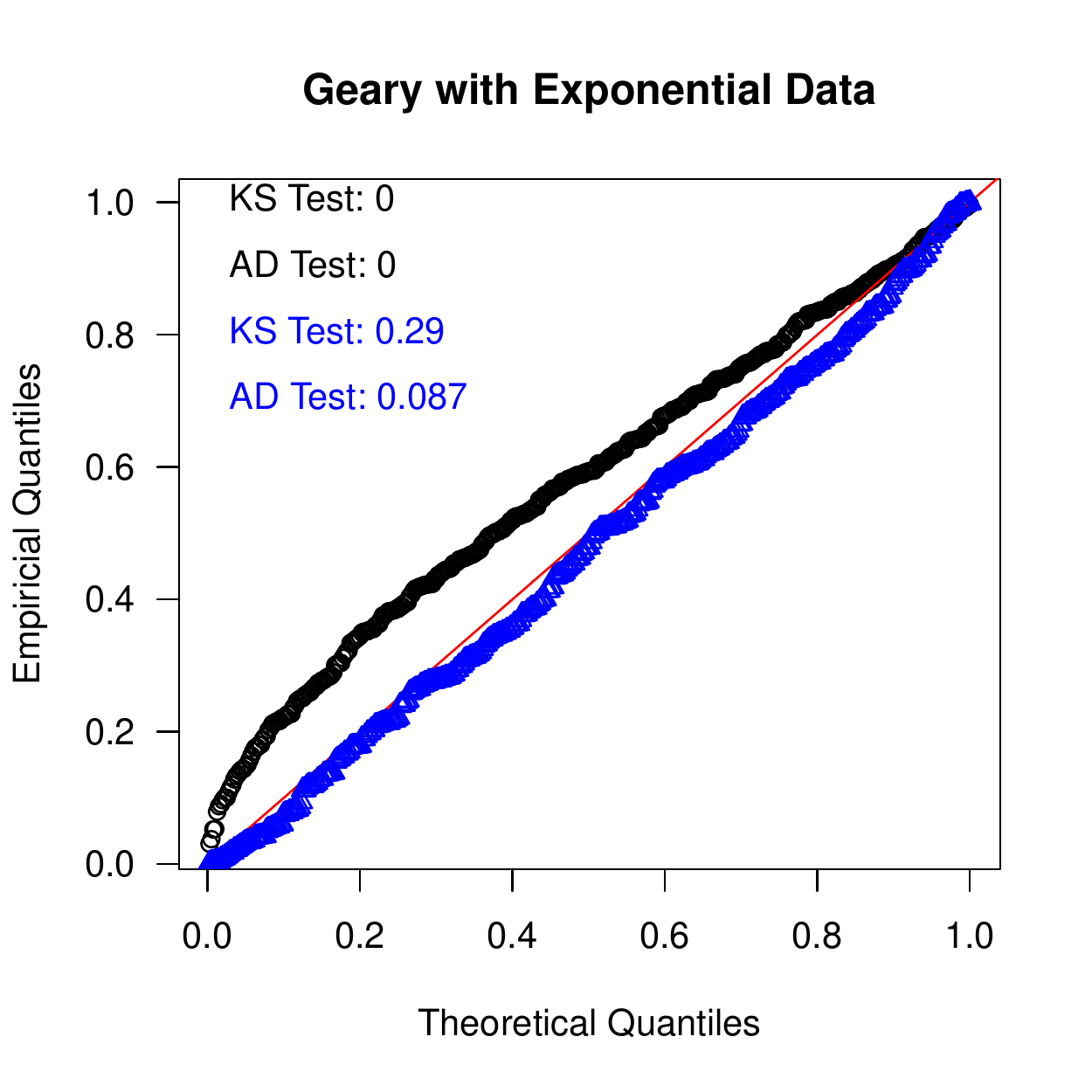}
	\end{center}
	\caption{
		\label{fig:glbSim}
		These plots depict of 400 p-values of ordered 
		p-values for global Moran's 
		and Geary's statistic and for independent Gaussian and 
		exponential data simulated on the entire map of Canada's 
		338 ridings.  This indicates that 
		Theorem~\ref{thm:globalGamma} produces p-values as would
		be expected in the null setting of independence expect 
		in the case of Geary's statistic with exponential data 
		where the blue triangles indicate empirically adjusted
		p-values.
	}
\end{figure}

In Figure~\ref{fig:glbPow}, we compare the statistical power of 
Theorem~\ref{thm:globalGamma} to the classic computation-based 
permutation test.  For both Moran's and Geary's statistic and 
both Gaussian and exponential data, we randomly generate 400 
datasets on the map of Canada with 6 different correlation matrices
being
$I + c_iA$ where $I$ is the identity matrix, $A$ is the adjacency
matrix for the map of Canada, and $c_i$ ranges from 0 to 0.15 for
Gaussian data and from 0 to 0.5 for exponential data.
The permutation test was performed by simulating 500 random 
permutations resulting
in a total of $500\times400\times6 = 1,200,000$ permutations in 
total. 
For Moran and Geary with Gaussian data, 
we see nearly identical statistical power from both 
methodologies.  
For Moran with Exponential data, there is a slight drop
in the statistical power.
In the case of Geary's statistic with
exponential data, we lose more
statistical power similar to the null setting above.  
By applying the empirical adjustment from Algorithm~\ref{algo:empBetaAlg}, 
we can recover some of the lost power.
As both Theorem~\ref{thm:localGamma} and~\ref{thm:globalGamma}
produce upper bounds the permutation test p-value, 
we note that the sharpness of these bounds is negatively
affected when the data is heavily skewed.  This is much 
more noticeable in Geary's statistic than in Moran's
statistic.

Lastly, we note that Theorem~\ref{thm:globalGamma} is 
specifically formulated to be a two-sided test.  Hence,
in Figure~\ref{fig:glbPow}, we are comparing its performance
with a two-sided permutation test.  As we only considered
positive correlations in this simulation, we could have 
achieved higher statistical power with a one-sided test.

\begin{figure}
	\begin{center}
		\includegraphics[width=0.45\textwidth]{\PICDIR/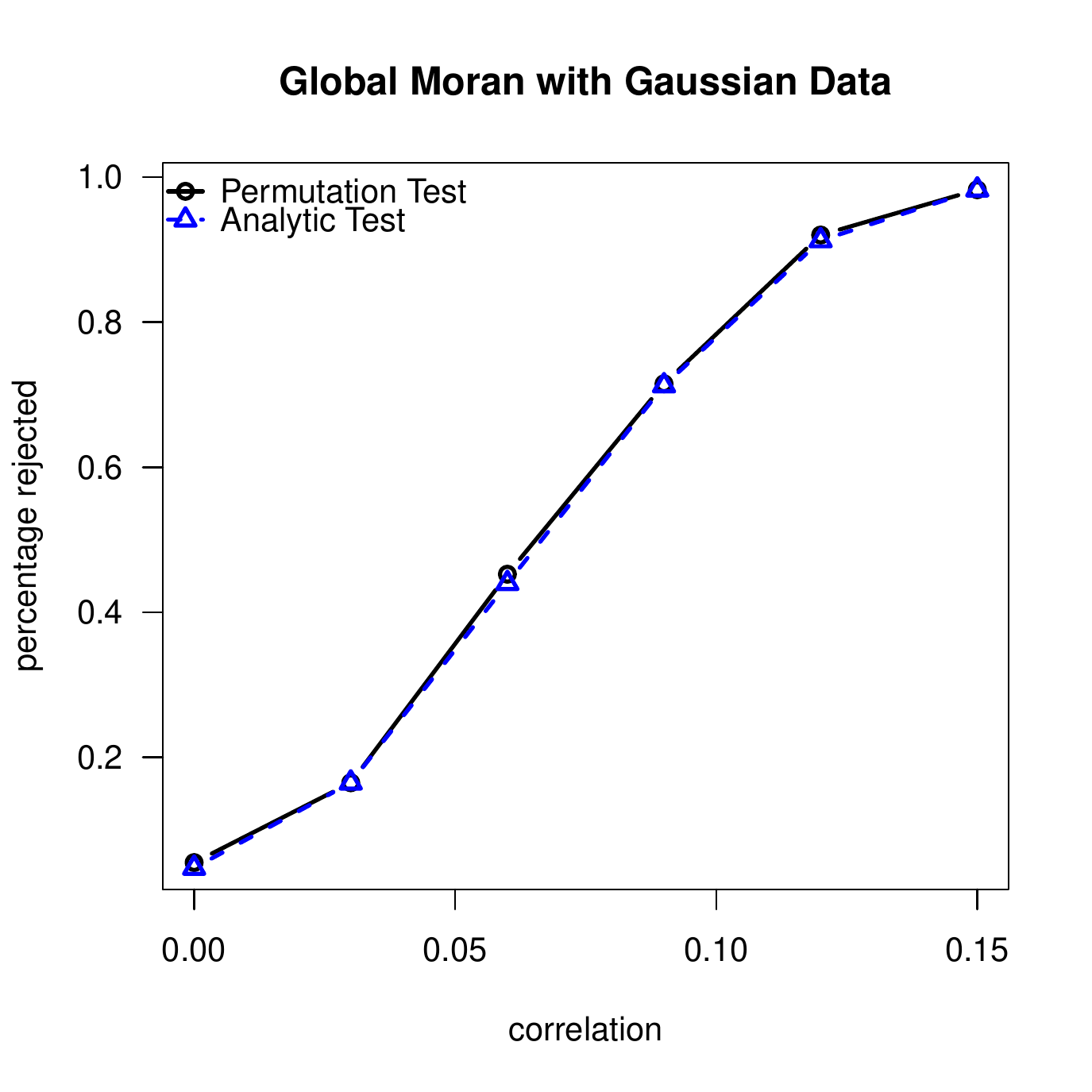} 
		\includegraphics[width=0.45\textwidth]{\PICDIR/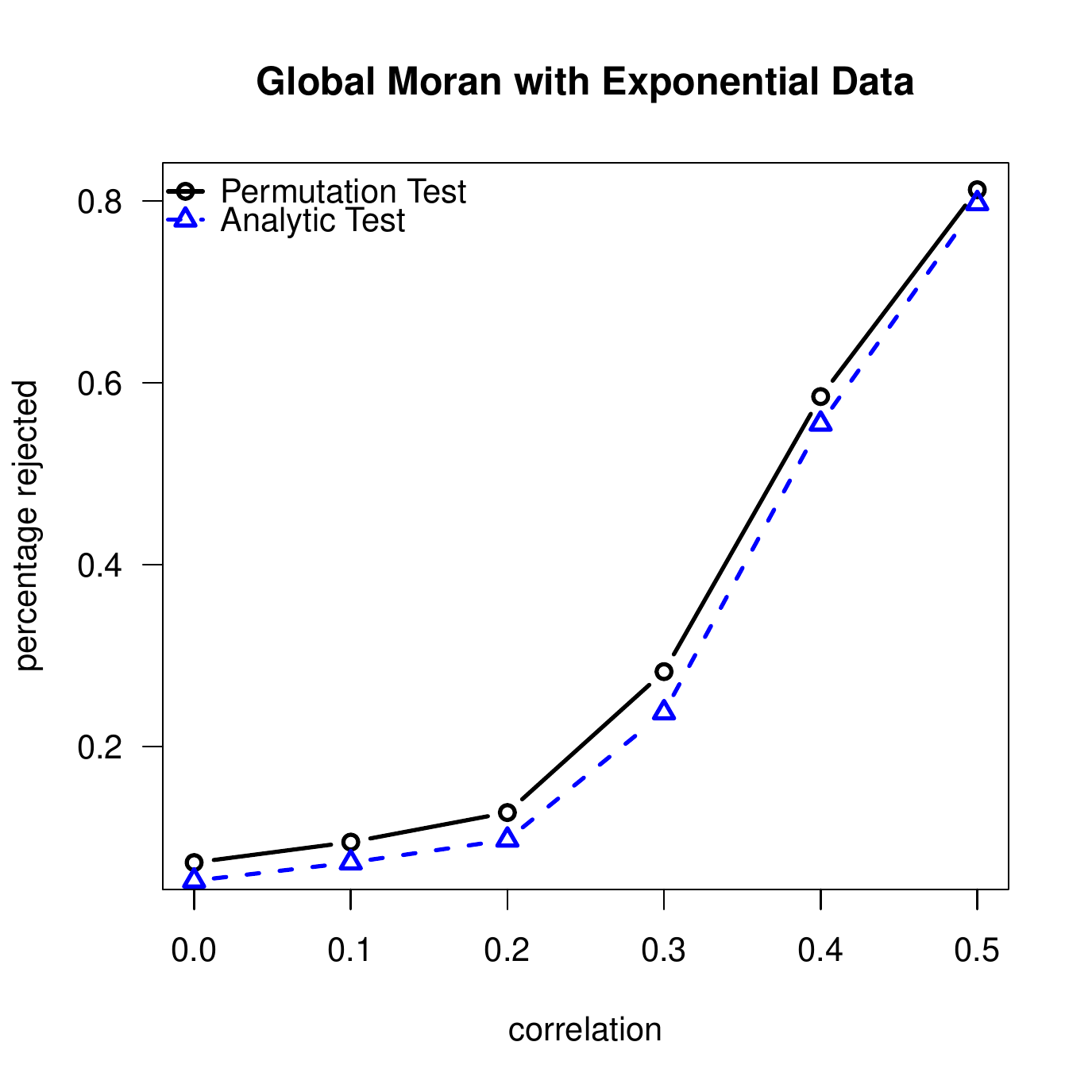}
		\includegraphics[width=0.45\textwidth]{\PICDIR/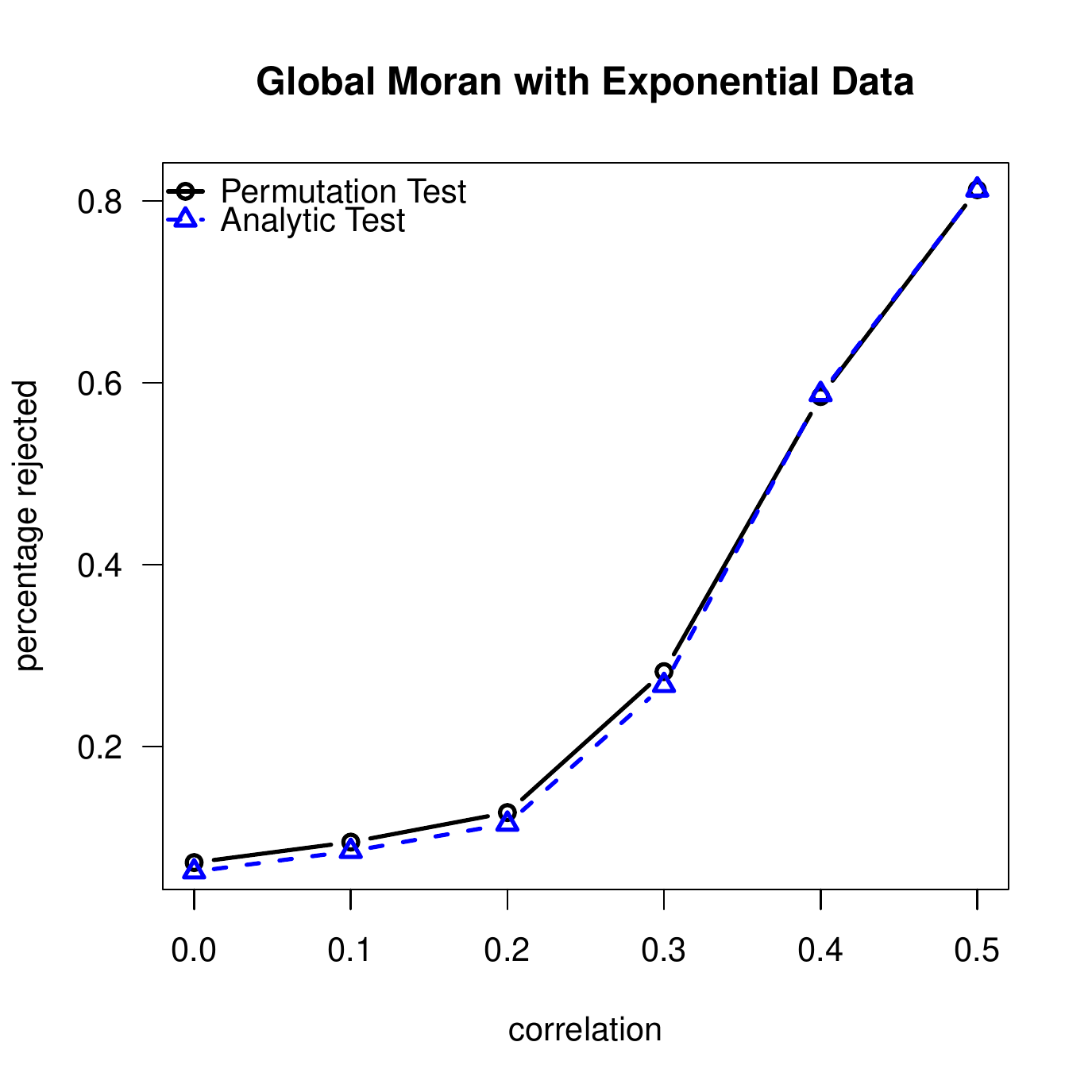}
		\includegraphics[width=0.45\textwidth]{\PICDIR/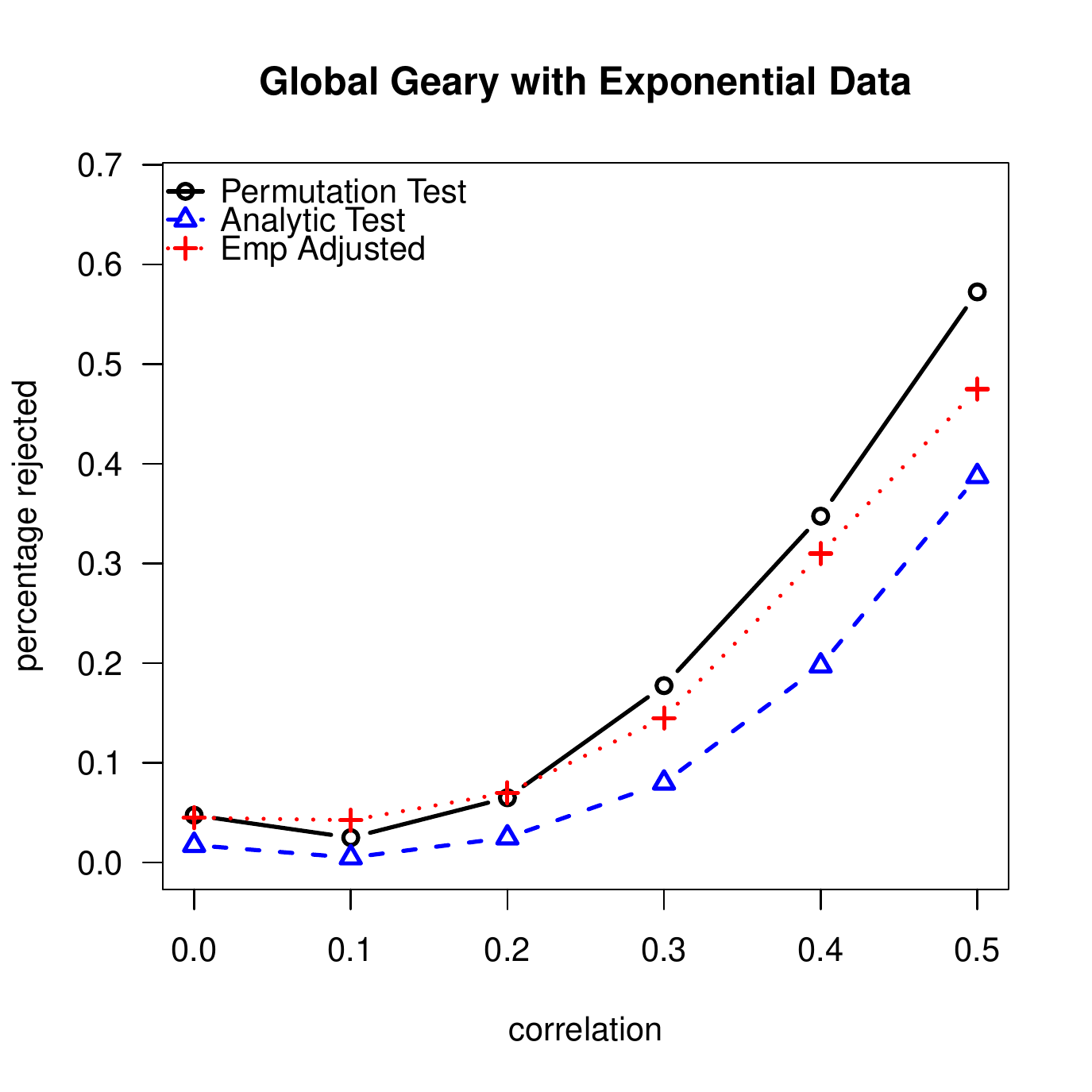}
	\end{center}
	\caption{
		\label{fig:glbPow}
		These plots compare the statistical power of our 
		method (blue $\triangle$) to the classic permutation test 
		(black $\circ$) with 500 permutations replicated 400 times
		at each of 6 different correlations.  The first three cases
		show that equivalent power is achieved in both methodologies.
		For Geary's statistic with skewed exponential data, the empirical
		adjustment is applied (red +'s) as in Figure~\ref{fig:glbSim}
		to recover lost power.
	}
\end{figure}

\subsection{Alberta Electorate Data}

The 2019 federal election 
resulted in all but one of Alberta's 34 ridings going to
the conservative party.  We apply our methodology to testing
for significant local spatial autocorrelation via both 
Moran's and Geary's statistics and for $k=1,2,3$ nearest 
neighbours weight matrices.  The measured response at 
each riding is the percentage of the popular vote 
captured by the conservative candidate.  Thus, our 
response variables $y_i\in[0,1]$ and left skewed 
as can be seen in Figure~\ref{fig:abHist}.
Thus, the assumption of normality does not hold.  
As the number of ridings is fixed at 34, 
we furthermore cannot rely on asymptotic statistics
as $n\nrightarrow\infty$.  In this section, we compare
p-values from our analytic variant on the permutation test to 
those from
the classic computation-based permutation test 
with 50,000 permutations at each node
We also consider p-values from z-scores 
based on the mean and variances computed in \cite{SOKAL1998}
in the appendix.
We note that while p-values based on the normal distribution
are common for Moran's statistic and readily available via the
\texttt{localmoran()} function in the \texttt{spdep}
\textsf{R} package 
\citep{BIVAND2013,BIVAND2018}, such a parametric approximation
is not advised for Geary's statistic \citep{ANSELIN2019,SEYAS2020}.

\begin{figure}
	\begin{center}
	  \includegraphics[width=0.65\textwidth]{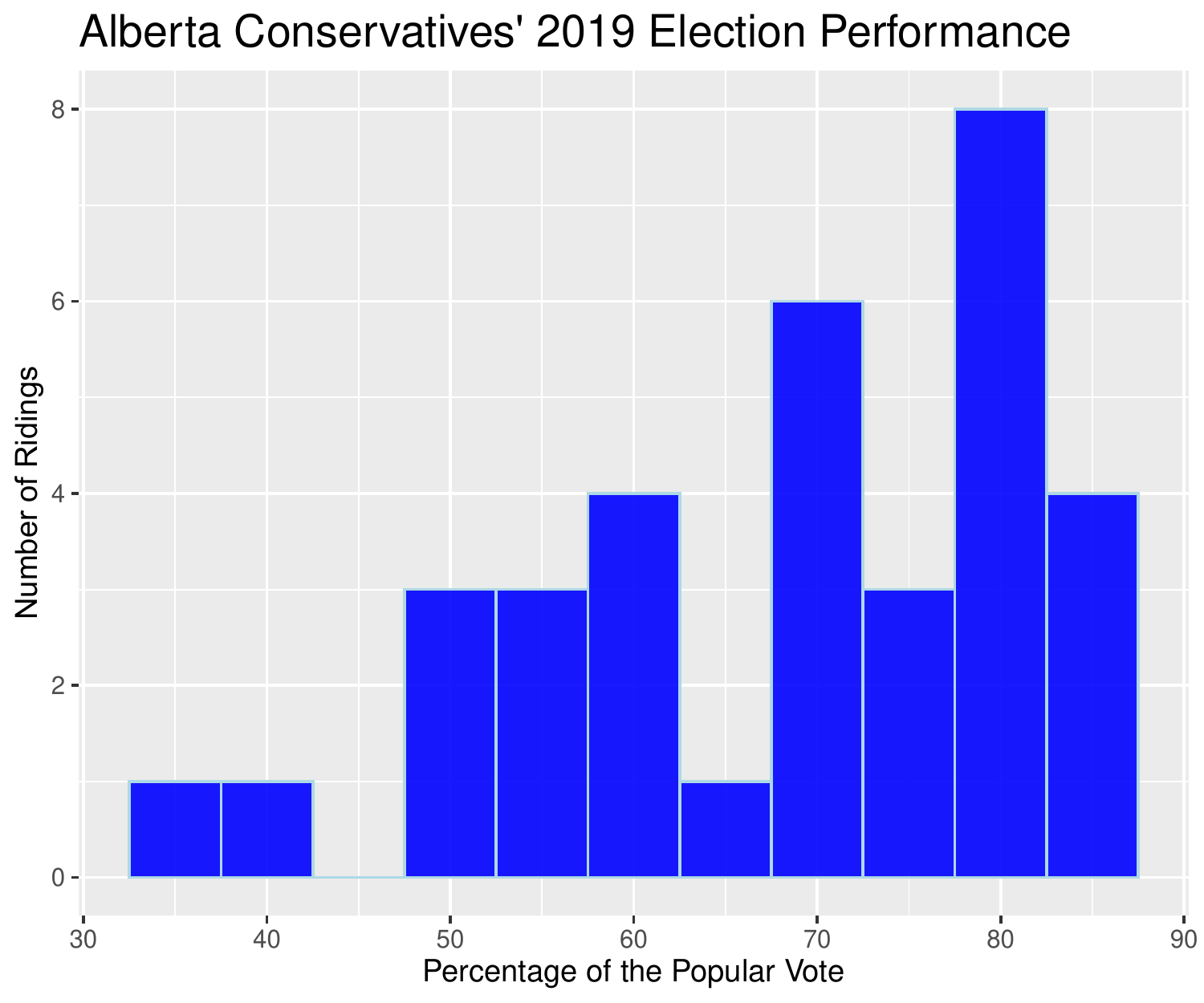}
	\end{center}
	\caption{
	  \label{fig:abHist}
	  A histogram of the popular vote gained by the
	  conservative candidates in the 34 ridings of Alberta.
	  The distribution of these measurements is clearly not
	  normal and is, in fact, skewed to the left.
	}
\end{figure}

We compare the p-value from the simulation-based permutation
test to the p-value produced from our analytic formulation 
of the permutation test.  For each vertex in the graph, 
50,000 permutations were randomly generated to compute
p-values for Moran's and Geary's statistics.  For choice
of weight matrix, we consider $k$-nearest neighbour matrices---i.e.
$W$ such that $w_{i,j}=1$ if the shortest path between
vertices $i$ and $j$ has length less than or equal to $k$
with $w_{i,i}=0$ for all $i$---for $k=1,2,3$.  When $k=3$,
two ridings are excluded from the analysis as they are within 
3 edges of all other ridings; these two ridings are 
\textit{Yellowhead} and \textit{Battle River--Crowfoot} located
in the centre west and centre east of the province, respectively.
The results are displayed in Figure~\ref{fig:pvComp}.  In 
general, our method returns similar if slightly more conservative
p-values than a computation-based permutation test.

\begin{figure}
  \begin{center}
  	\includegraphics[width=0.45\textwidth]{\PICDIR/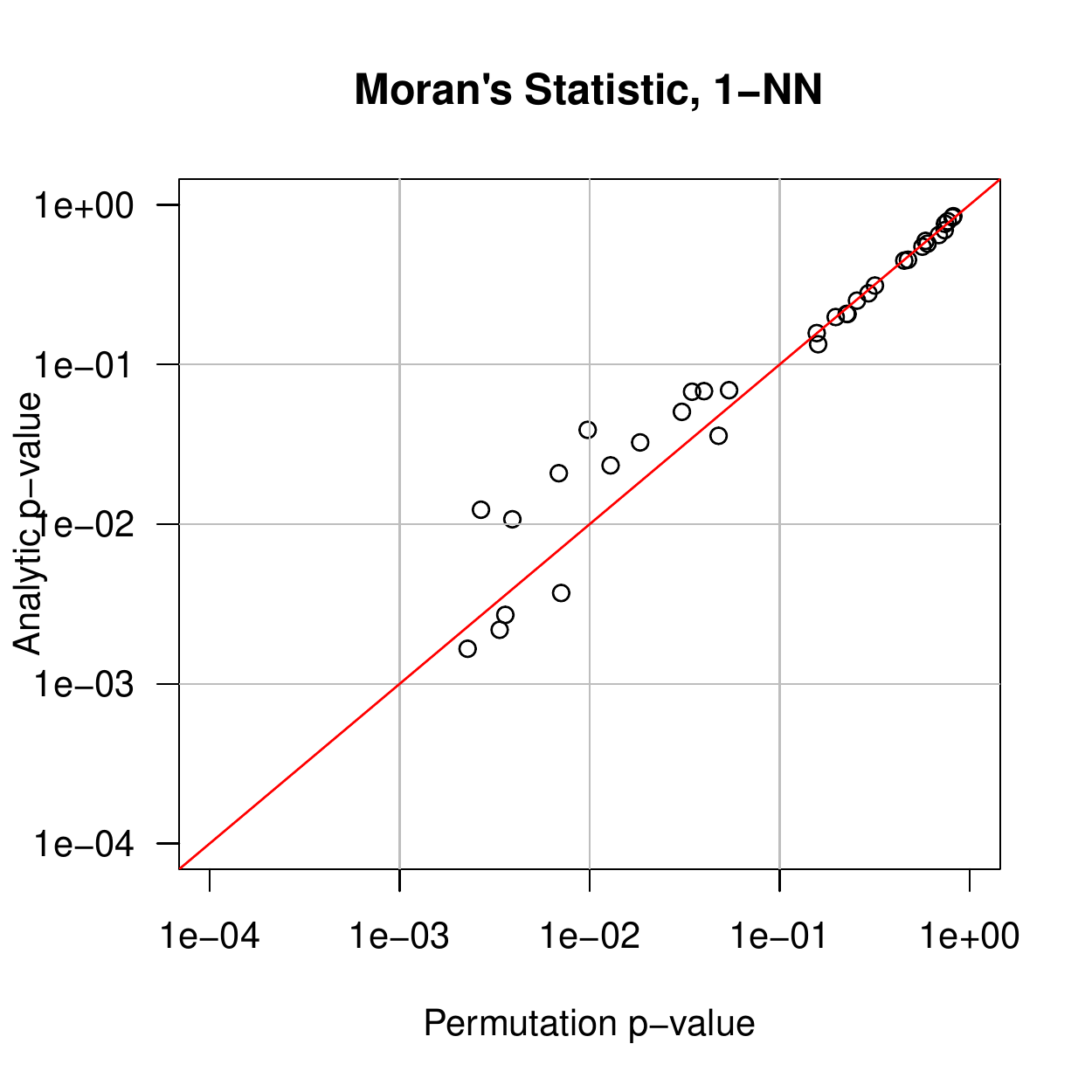}
  	\includegraphics[width=0.45\textwidth]{\PICDIR/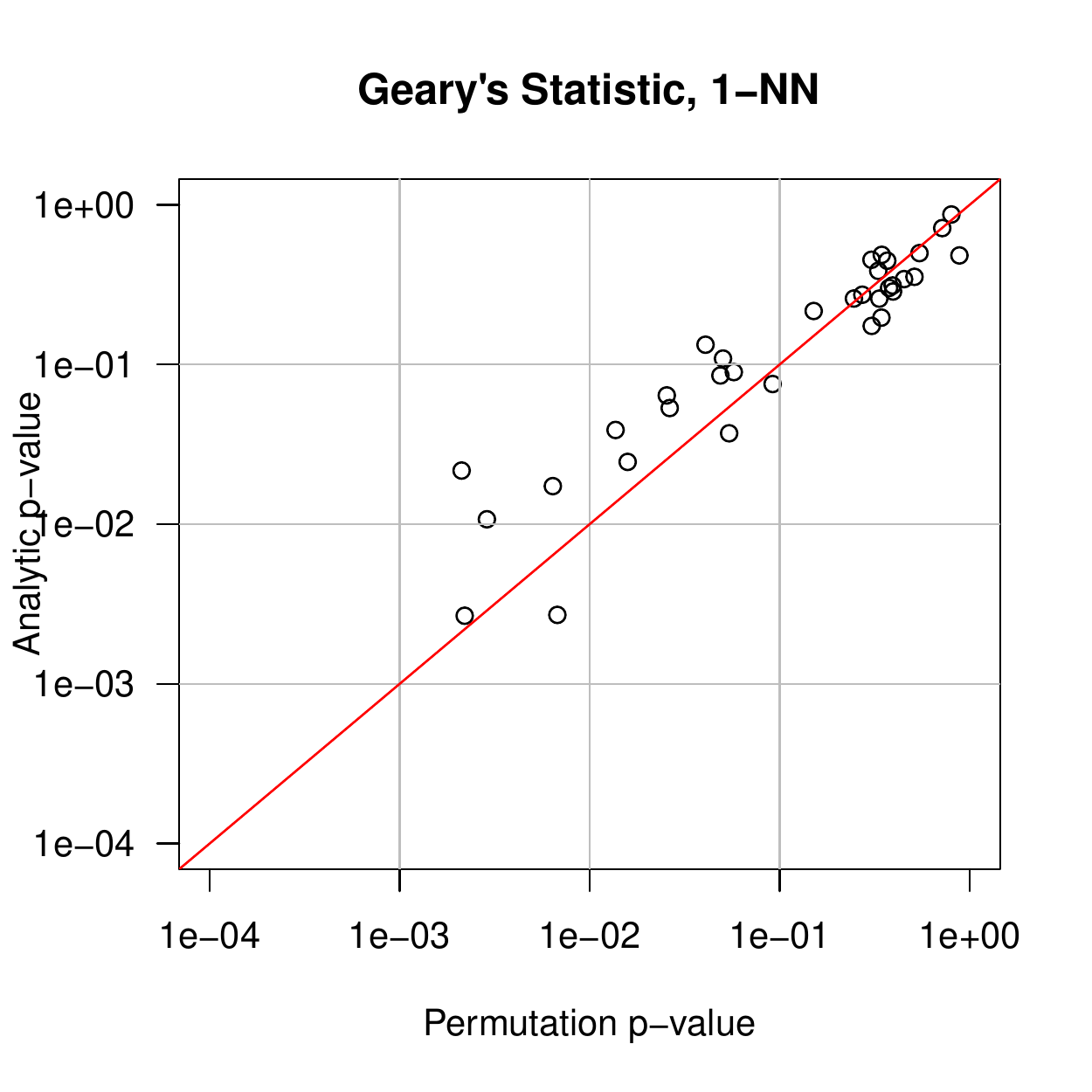}
  	\includegraphics[width=0.45\textwidth]{\PICDIR/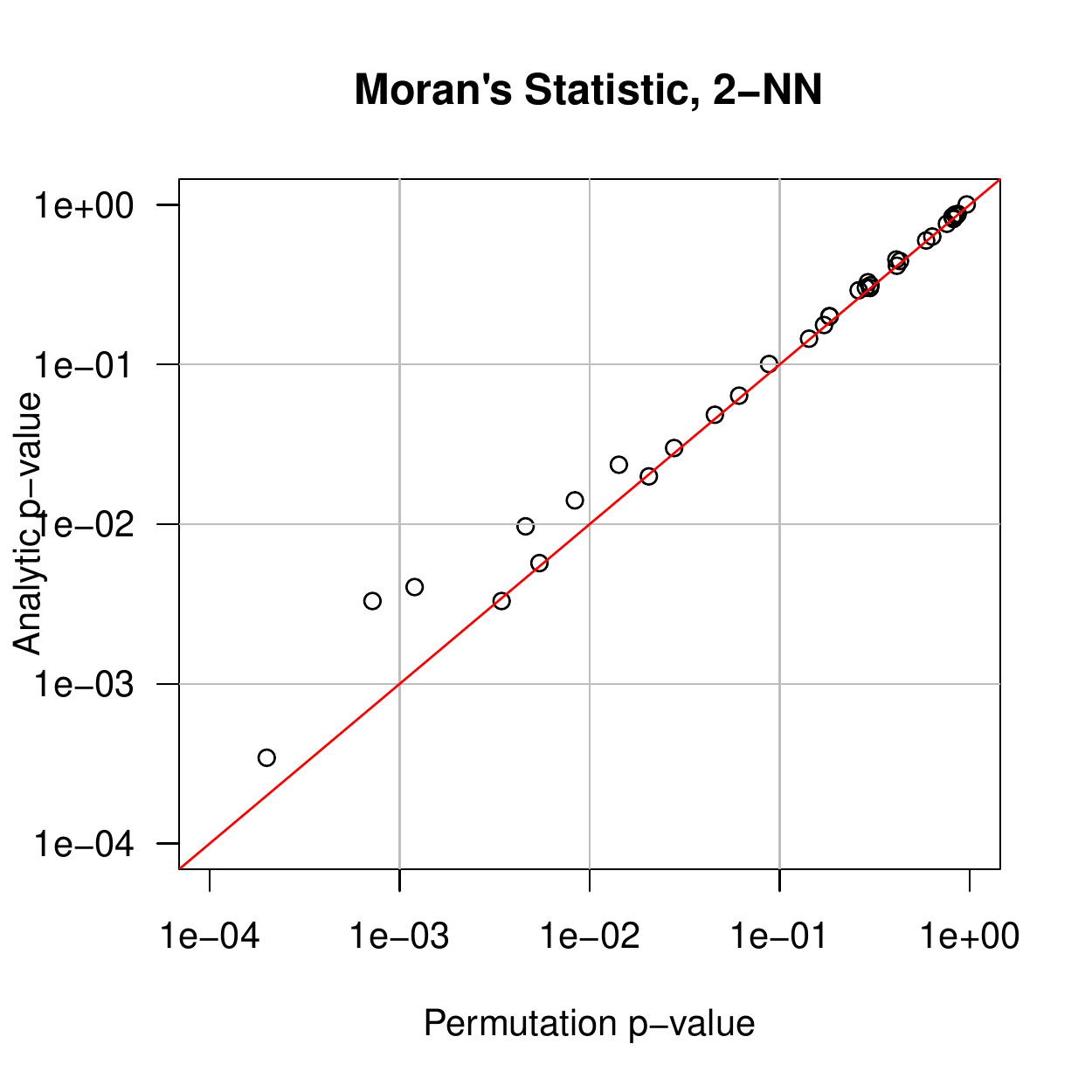}
    \includegraphics[width=0.45\textwidth]{\PICDIR/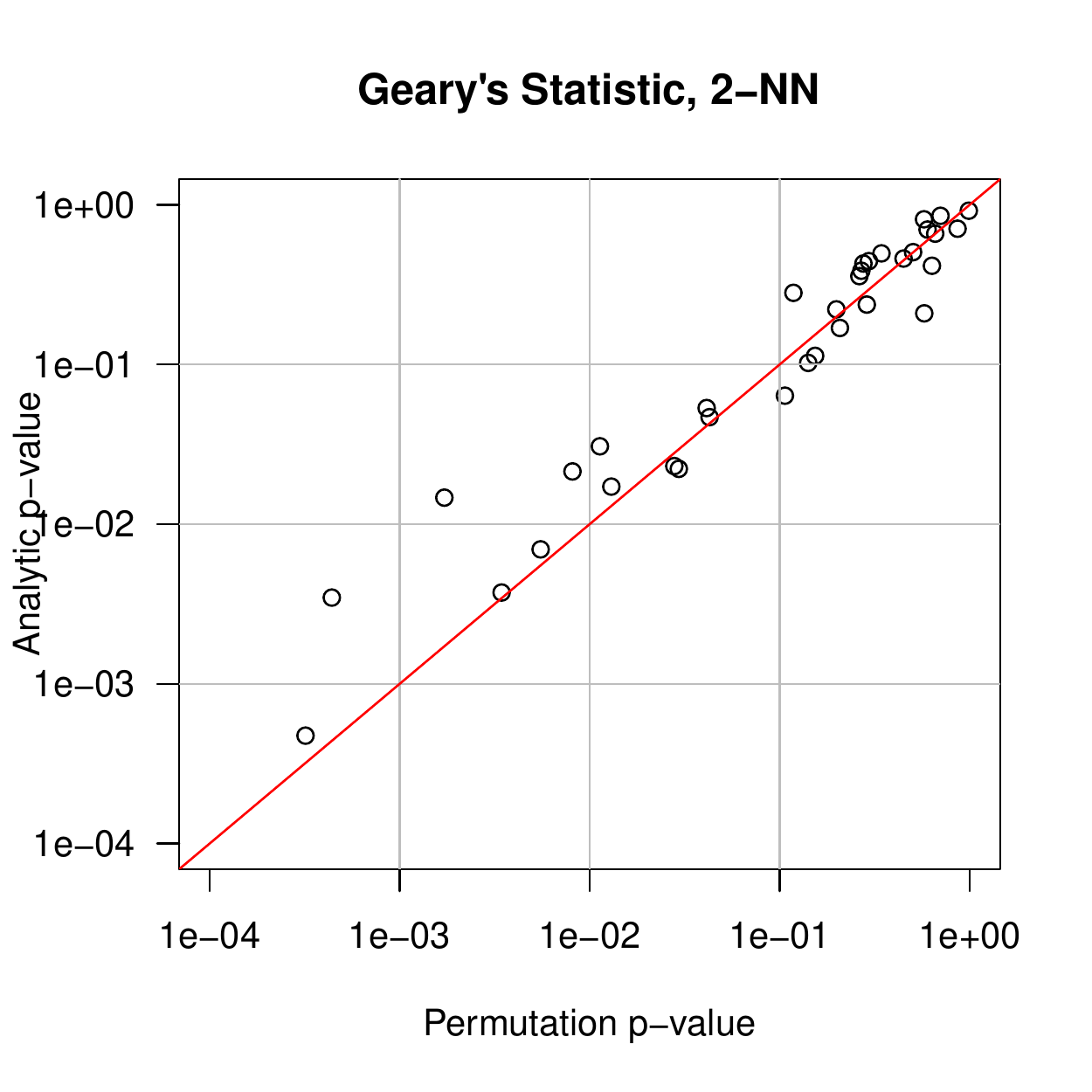}
  	\includegraphics[width=0.45\textwidth]{\PICDIR/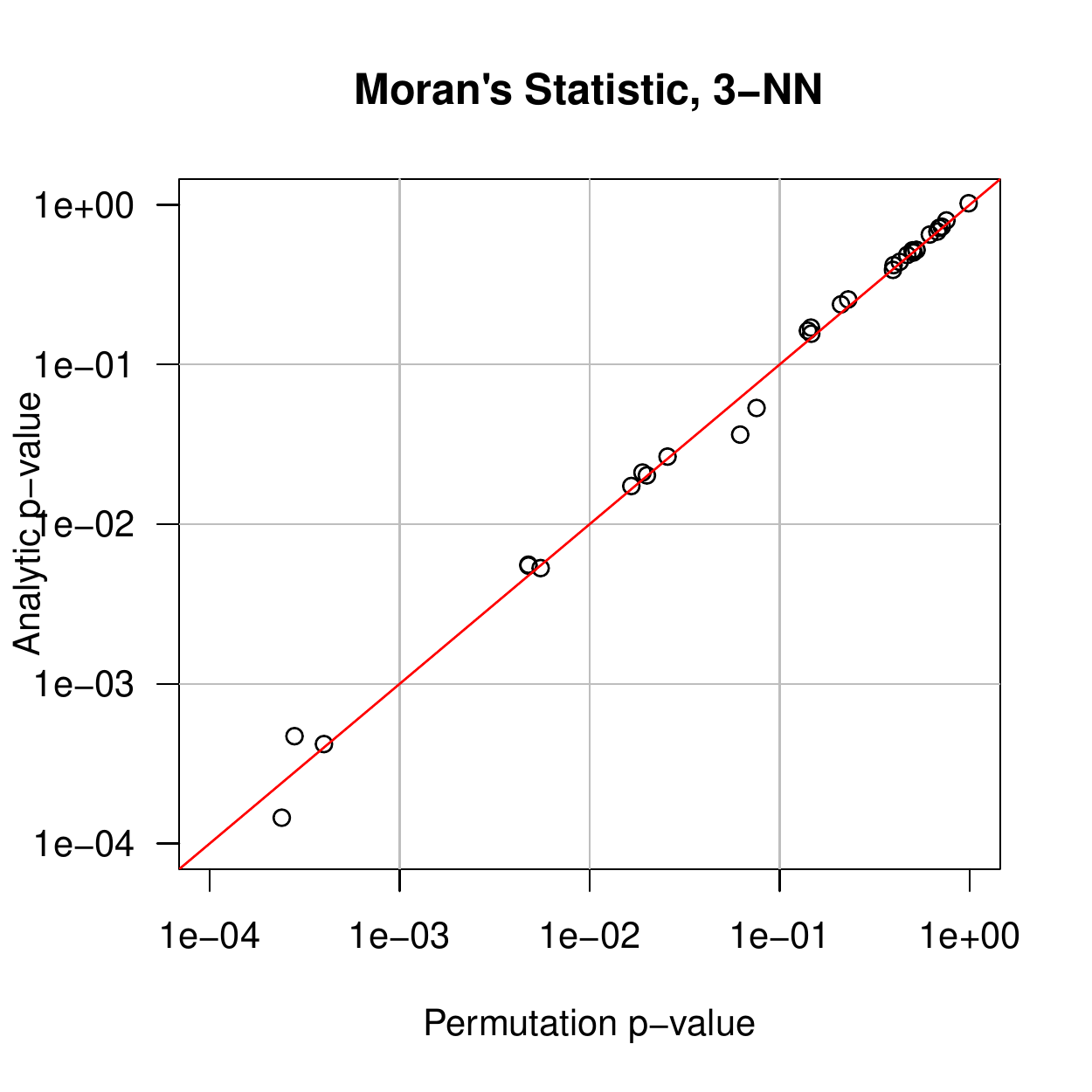}
    \includegraphics[width=0.45\textwidth]{\PICDIR/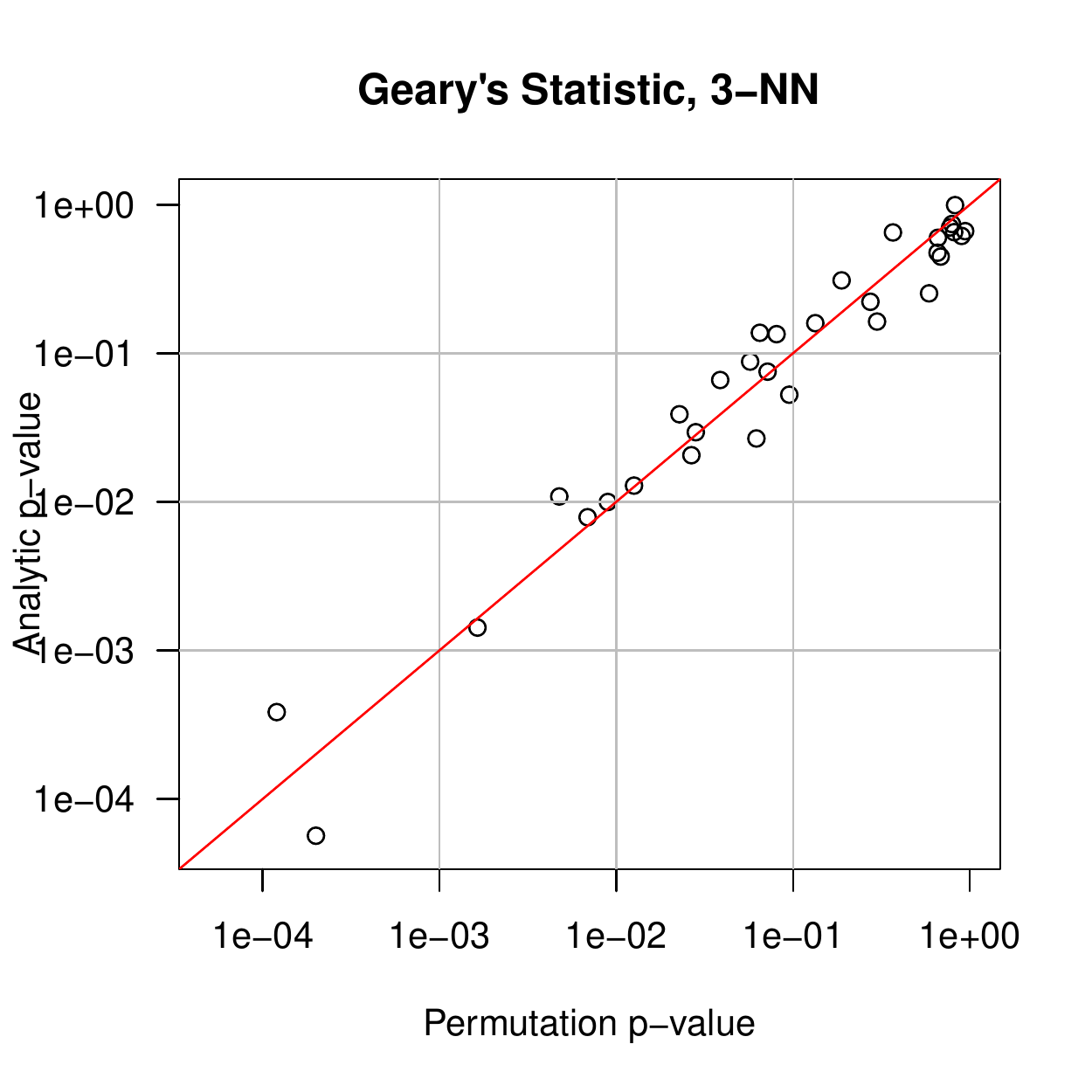}
  \end{center}
  \caption{
  	\label{fig:pvComp}
  	A comparison of the p-values produced by a simulation-based 
  	permutation test with 50,000 permutations per vertex
  	and the p-values produced by our analytic variant of
  	the permutation test.  The left column considers
  	Moran's statistic; the right column considers Geary's
  	statistic.  The three rows from top to bottom consider
  	the 1, 2, and 3-nearest neighbours weight matrix, respectively.
  }
\end{figure}

\begin{figure}
  \begin{center}
  	\underline{\bf One-Nearest-Neighbours}
  	\includegraphics[width=0.9\textwidth]{\PICDIR/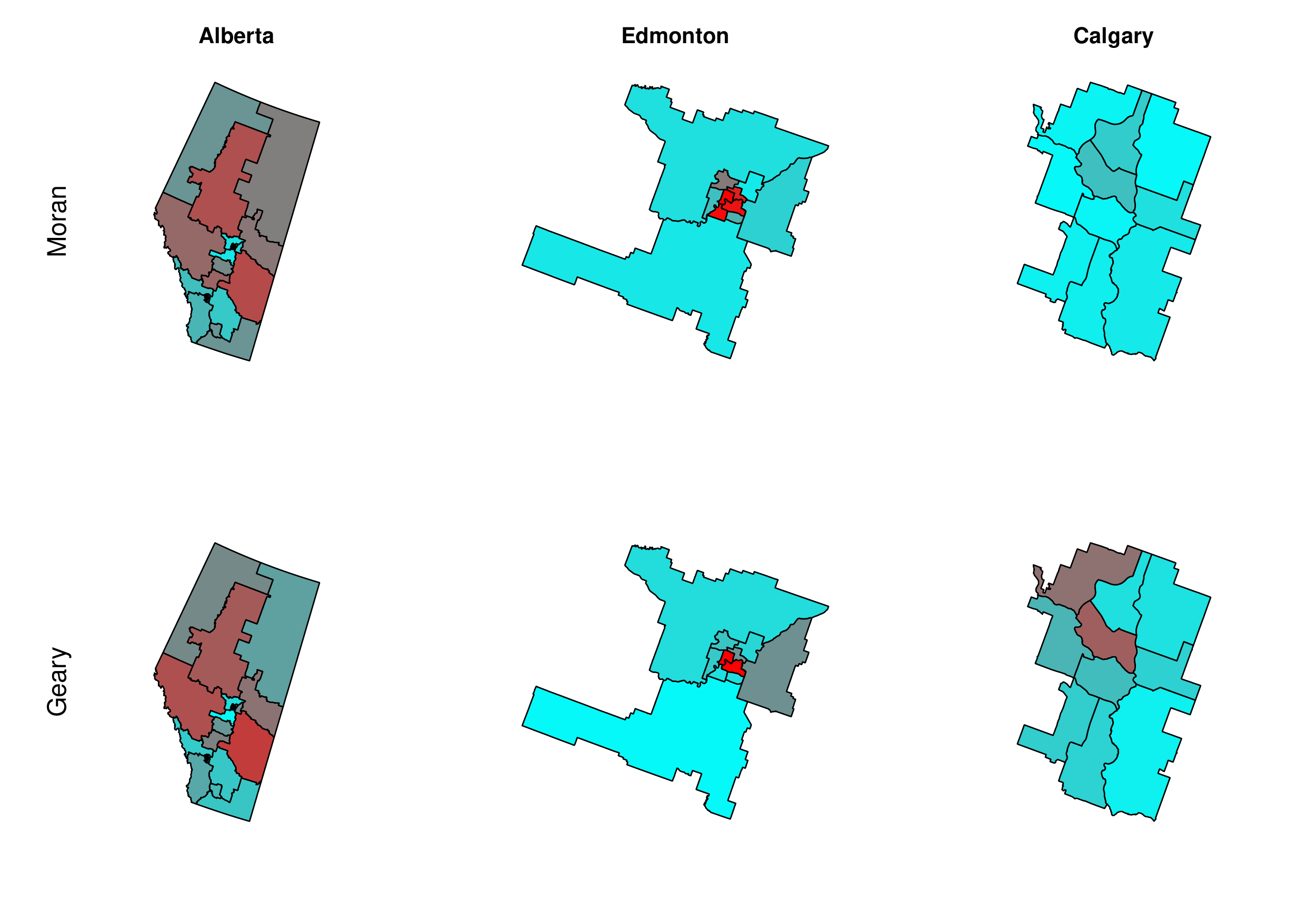}
  	\underline{\bf Two-Nearest-Neighbours}
  	\includegraphics[width=0.9\textwidth]{\PICDIR/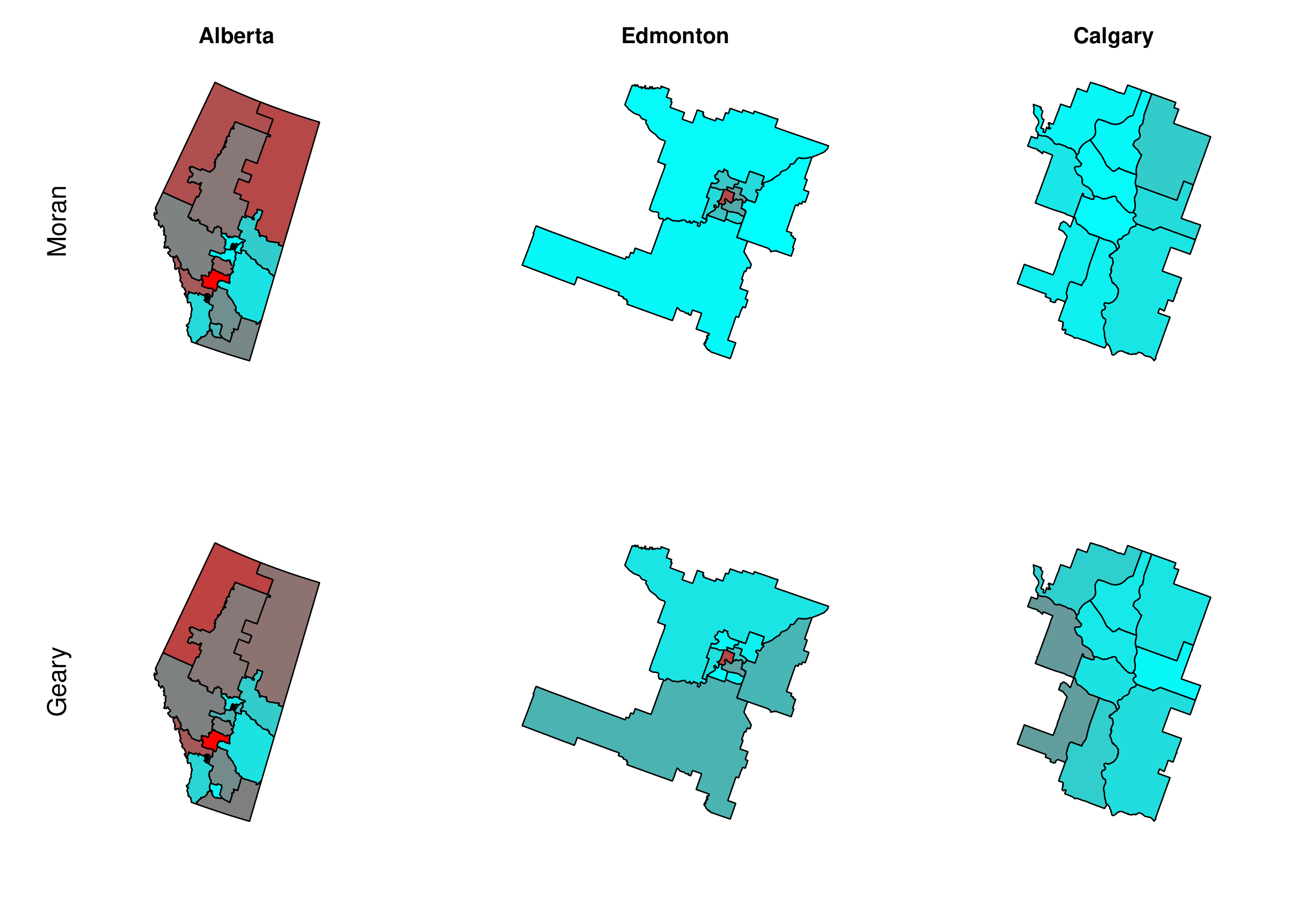}
  \end{center}
  \caption{
  	\label{fig:abMaps}
  	Maps of the province of Alberta along with the cities 
  	of Edmonton and Calgary coloured by p-values produced 
  	by local Moran's (upper row) and local Geary's (lower row) 
  	statistics using Theorem~\ref{thm:localGamma} using the 
  	one and two nearest neighbours matrices for $W$. 
  	Red indicates
  	a small p-value near zero while cyan indicates a large
  	p-value near 1.
  }
\end{figure}

\section{Discussion and Future Extensions}

While our focus in this work was on real valued measurements,
this methodology could be extended to the settings of multivariate
or functional responses. The theorems we rely on
from \cite{KASHLAK_KHINTCHINE2020} extend from the real line to 
general Banach spaces.
Spatial modelling with functional 
responses was considered, for example, in \citep{TAVAKOLI2019}
and is a challenging problem.

Our main theorem is proven for binary weight matrices.  
The 0's and 1's naturally lead to test statistics
being reformulated into two sample tests so that the ideas 
of \cite{KASHLAK_KHINTCHINE2020} can be applied.  However, 
non-binary weight matrices are often of interest in spatial
data analysis.  Extensions to such will require the adaptation
of novel versions of the Khintchine inequality such as 
\cite{HAVRILLA2019}.

\section*{Acknowledgements}

The authors would like to thank the 
Natural Sciences and Engineering Research Council of Canada 
(NSERC) for the funding provided via their Discovery Grant
program.


\bibliographystyle{plainnat}
\bibliography{./kasharticle,./kashbook,./kashpack,./kashself}

\begin{thebibliography}{22}
\providecommand{\natexlab}[1]{#1}
\providecommand{\url}[1]{\texttt{#1}}
\expandafter\ifx\csname urlstyle\endcsname\relax
  \providecommand{\doi}[1]{doi: #1}\else
  \providecommand{\doi}{doi: \begingroup \urlstyle{rm}\Url}\fi

\bibitem[Abramowitz and Stegun(1972)]{ABRAMOWITZ_STEGUN}
Milton Abramowitz and Irene~A Stegun.
\newblock Handbook of mathematical functions with formulas, graphs, and
  mathematical tables, 1972.

\bibitem[Anselin(1995)]{ANSELIN1995}
Luc Anselin.
\newblock Local indicators of spatial association—lisa.
\newblock \emph{Geographical analysis}, 27\penalty0 (2):\penalty0 93--115,
  1995.

\bibitem[Anselin(2019)]{ANSELIN2019}
Luc Anselin.
\newblock A local indicator of multivariate spatial association: extending
  geary's c.
\newblock \emph{Geographical Analysis}, 51\penalty0 (2):\penalty0 133--150,
  2019.

\bibitem[Bivand and Wong(2018)]{BIVAND2018}
Roger~S Bivand and David~WS Wong.
\newblock Comparing implementations of global and local indicators of spatial
  association.
\newblock \emph{Test}, 27\penalty0 (3):\penalty0 716--748, 2018.

\bibitem[Bivand et~al.(2013)Bivand, Pebesma, and Gomez-Rubio]{BIVAND2013}
Roger~S. Bivand, Edzer Pebesma, and Virgilio Gomez-Rubio.
\newblock \emph{Applied spatial data analysis with {R}, Second edition}.
\newblock Springer, NY, 2013.
\newblock URL \url{http://www.asdar-book.org/}.

\bibitem[Cliff and Ord(1981)]{CLIFFORD1981}
Andrew~David Cliff and J~Keith Ord.
\newblock \emph{Spatial processes: models \& applications}.
\newblock Taylor \& Francis, 1981.

\bibitem[DLMF(2020)]{NIST:DLMF}
DLMF.
\newblock {\it NIST Digital Library of Mathematical Functions}.
\newblock http://dlmf.nist.gov/, Release 1.0.28 of 2020-09-15, 2020.
\newblock URL \url{http://dlmf.nist.gov/}.
\newblock F.~W.~J. Olver, A.~B. {Olde Daalhuis}, D.~W. Lozier, B.~I. Schneider,
  R.~F. Boisvert, C.~W. Clark, B.~R. Miller, B.~V. Saunders, H.~S. Cohl, and
  M.~A. McClain, eds.

\bibitem[Doman(1996)]{DOMAN1996}
B~Doman.
\newblock An asymptotic expansion for the incomplete beta function.
\newblock \emph{Mathematics of computation}, 65\penalty0 (215):\penalty0
  1283--1288, 1996.

\bibitem[Gaetan and Guyon(2010)]{GAETAN2010}
Carlo Gaetan and Xavier Guyon.
\newblock \emph{Spatial statistics and modeling}, volume~90.
\newblock Springer, 2010.

\bibitem[Garling(2007)]{GARLING2007}
David~JH Garling.
\newblock \emph{Inequalities: a journey into linear analysis}.
\newblock Cambridge University Press, 2007.

\bibitem[Getis and Ord(2010)]{GETISORDb}
Arthur Getis and J~Keith Ord.
\newblock The analysis of spatial association by use of distance statistics.
\newblock In \emph{Perspectives on spatial data analysis}, pages 127--145.
  Springer, 2010.

\bibitem[Haagerup(1981)]{HAAGERUP1981}
Uffe Haagerup.
\newblock The best constants in the {Khintchine} inequality.
\newblock \emph{Studia Mathematica}, 70:\penalty0 231--283, 1981.

\bibitem[Havrilla and Tkocz(2019)]{HAVRILLA2019}
Alex Havrilla and Tomasz Tkocz.
\newblock Sharp khinchin-type inequalities for symmetric discrete uniform
  random variables.
\newblock \emph{arXiv preprint arXiv:1912.13345}, 2019.

\bibitem[Herscovici and Spektor(2020)]{SUSANNAORLI2020}
Orli Herscovici and Susanna Spektor.
\newblock The best constant in the {K}hinchine inequality for slightly
  dependent random variables.
\newblock \emph{arXiv preprint arXiv:1806.03562}, 2020.

\bibitem[Hubert(1985)]{HUBERT1985}
Lawrence~J Hubert.
\newblock Combinatorial data analysis: association and partial association.
\newblock \emph{Psychometrika}, 50\penalty0 (4):\penalty0 449--467, 1985.

\bibitem[Kashlak et~al.(2020)Kashlak, Myroshnychenko, and
  Spektor]{KASHLAK_KHINTCHINE2020}
Adam~B Kashlak, Sergii Myroshnychenko, and Susanna Spektor.
\newblock Analytic permutation testing via {K}ahane--{K}hintchine inequalities.
\newblock \emph{arXiv preprint arXiv:2001.01130}, 2020.

\bibitem[Mantel(1967)]{MANTEL1967}
Nathan Mantel.
\newblock The detection of disease clustering and a generalized regression
  approach.
\newblock \emph{Cancer research}, 27\penalty0 (2 Part 1):\penalty0 209--220,
  1967.

\bibitem[Seya(2020)]{SEYAS2020}
Hajime Seya.
\newblock Global and local indicators of spatial associations.
\newblock In \emph{Spatial Analysis Using Big Data}, pages 33--56. Elsevier,
  2020.

\bibitem[Sokal et~al.(1998)Sokal, Oden, and Thomson]{SOKAL1998}
Robert~R Sokal, Neal~L Oden, and Barbara~A Thomson.
\newblock Local spatial autocorrelation in a biological model.
\newblock \emph{Geographical Analysis}, 30\penalty0 (4):\penalty0 331--354,
  1998.

\bibitem[Spektor(2016)]{SPEKTOR2016}
Susanna Spektor.
\newblock Restricted {Khinchine} inequality.
\newblock \emph{Canadian Mathematical Bulletin}, 59\penalty0 (1):\penalty0
  204--210, 2016.

\bibitem[Tavakoli et~al.(2019)Tavakoli, Pigoli, Aston, and
  Coleman]{TAVAKOLI2019}
Shahin Tavakoli, Davide Pigoli, John~AD Aston, and John~S Coleman.
\newblock A spatial modeling approach for linguistic object data: Analyzing
  dialect sound variations across great britain.
\newblock \emph{Journal of the American Statistical Association}, 114\penalty0
  (527):\penalty0 1081--1096, 2019.

\bibitem[Waller and Gotway(2004)]{WALLER2004}
Lance~A Waller and Carol~A Gotway.
\newblock \emph{Applied spatial statistics for public health data}, volume 368.
\newblock John Wiley \& Sons, 2004.

\end{thebibliography}

\appendix

\section{Proofs}

\subsection{LISA Proofs}

\begin{proof}[Proof of Theorem~\ref{thm:localGamma}]
	We first recall that $n$ is the total number of vertices in 
	the graph $\mathcal{G}$ and that $m_i$ is the number of edges at vertex
	$\nu_i$.  For binary weights $w_{i,j}\in\{0,1\}$, we apply 
	an affine transformation to define 
	imbalanced Rademacher weights 
	$$
	\delta_{i,j} = \frac{n-1}{m_i(n-m_i-1)} w_{i,j}-\frac{1}{n-m_i-1}\in\left\{
	-\frac{1}{n-m_i-1},\frac{1}{m_i}
	\right\}
	$$
	for $i\ne j$ while maintaining that $\delta_{i,i}=w_{i,i}=0$.
	Note that since $\sum_jw_{i,j}=m_i$ that $\sum_j \delta_{i,j}=0$.
	
	Considering $\gamma_i = \sum_{j=1}^nw_{i,j}\lmb_{i,j}$, 
	we have that 
	\begin{align*}
	\gamma_i = \sum_{j=1}^nw_{i,j}\lmb_{i,j}
	&= \frac{m_i(n-m_i-1)}{n-1}\sum_{j=1}^n\left\{
	\delta_{i,j}+(n-m_i-1)^{-1}
	\right\}\lmb_{i,j}\\
	&= \frac{m_i(n-m_i-1)}{n-1}\sum_{j=1}^n\delta_{i,j}\lmb_{i,j} + m_i\bar{\lmb}_{-i}
	\end{align*}
	where $\bar{\lmb}_{-i} = (n-1)^{-1}\sum_{j\ne i}\lmb_{i,j}$.
	Therefore, our statistic $\gamma_i$ is equivalent up to affine transformation
	to a two sample test for equality of the mean
	of the $\lmb_{i,j}$ such that $w_{i,j}=1$ and the mean of those $\lmb_{i,j}$
	with $w_{i,j}=0$ excluding the value $\lmb_{i,i}$ from this test.
	Thus, for $\pi$ being a uniformly random element of $\mathbb{S}_n$, 
	the symmetric group on $n$ elements, with the restriction that 
	$\pi(i)=i$, we define the permuted test statistic to be 
	$
	\gamma_i(\pi) =
	\frac{m_i(n-m_i-1)}{n-1}\sum_{j=1}^n\delta_{i,j}\lmb_{i,\pi(j)}+m_i\bar{\lmb}_{-i}
	$.  We note that a permutation test on $\gamma_i$ is equivalent to 
	a permutation test on $T_i=\sum_{j=1}^n\delta_{i,j}\lmb_{i,j}$. 
	Let $\Omega_i = [
	\delta\in\{-\frac{1}{n-m_i-1},0,\frac{1}{m_i}\}^{n}\,|\,
	\delta_{i}=0,\delta_{j}\ne0\text{ for }j\ne i,
	\sum_{j=1}^n\delta_j=0
	]$ be the set of possible $n$-dimensional weight vectors $\delta$
	that fix $\delta_i=0$.
	We note the following correspondence similar to \cite{SPEKTOR2016} that 
	\begin{multline*}
	\left\{\pi\in\mathbb{S}_n\middle|\pi(i)=i\right\} 
	\leftrightarrow 
	\left[
	\delta\in\Omega_i\,\middle|\,\right.\\\left.
	\delta_i=0\text{ and }\left\{
	\begin{array}{ll}
	\text{for }i\le m_i, &
	\delta_j=\frac{1}{m_i}\text{ if }\pi(j)\le m_i+1
	\text{ and }\delta_j=-\frac{1}{n-m_i-1}\text{ if  }\pi(j)>m_i+1 \\
	\text{for }i> m_i, &
	\delta_j=\frac{1}{m_i}\text{ if }\pi(j)\le m_i
	\text{ and }\delta_j=-\frac{1}{n-m_i-1}\text{ if }\pi(j)>m_i
	\end{array}
	\right\}
	\right].
	\end{multline*}
	Thus, we can consider the permuted test statistic with respect to a dependent
	vector of random weights $\delta\in\Omega$.
	That is, conditional of the 
	$y_i$, $T_i(\pi)$ can be treated as a weakly dependent weighted Rademacher sum. 
	Applying Theorem~2.1 of \cite{KASHLAK_KHINTCHINE2020} for imbalanced two sample tests 
	for equality of means under Condition~\ref{cond:noConVert} that
	$m_i \le n-m_i-1$, 
	we have that
	$$
	\prob{ \abs{T_i(\pi)} \ge t } \le
	\exp\left( 
	- \frac{m_i^3t^2}{2s_i^2(n-1)^2}
	\right)
	$$
	where $s_i^2=(n-1)^{-1}\sum_{j\ne i}(\lmb_{i,j}-\bar{\lmb}_{-i})^2$ is the 
	sample variance of the $\lmb_{i,j}$ for $j\ne i$.  
	Translating back to the local gamma index, we have 
	\begin{align*}
	\prob{ \abs{\gamma_i(\pi)-m_i\bar{\lmb}_{-i}} \ge \gamma_i } &\le
	\exp\left( 
	-\frac{m_i^3}{2s_i^2(n-1)^2}\left[
	\frac{n-1}{m_i(n-m_i-1)}\gamma_i 
	\right]^2
	\right)\\
	&\le
	\exp\left( 
	-\frac{m_i\gamma_i^2}{2s_i^2(n-m_i-1)^2}
	\right)
	\end{align*}

	For the final part of Theorem~\ref{thm:localGamma}, we
	apply the beta transform from \cite{KASHLAK_KHINTCHINE2020} Proposition~2.5
	resulting in 
	$$
	\prob{ \abs{\gamma_i(\pi)-m_i\bar{\lmb}_{-i}} \ge \gamma_i } \le 
	C_0 I\left[
	\exp\left( 
	-\frac{m_i\gamma_i^2}{2s_i^2(n-m_i-1)^2}
	\right);
	\frac{(n-1)(n-m_i-1)}{m_i^2},\frac{1}{2}
	\right]
	$$
	where $I[\cdot]$ is the regularized incomplete beta function and
	$$
	C_0 = \frac{
		{\sqrt{(n-1)(n-m_i-1)}}\Gamma\left(\frac{(n-1)(n-m_i-1)}{m_i^2}\right)
	}{
		m_i\Gamma\left(\frac{1}{2}+\frac{(n-1)(n-m_i-1)}{m_i^2}\right)
	}
	$$
	with $\Gamma(\cdot)$ the gamma function.  
\end{proof}

\subsection{GISA Proofs}

\begin{lemma}
  \label{lem:halfBinom}
  For $q>0$ and $\abs{c}<1$
  \begin{equation}
  \label{eqn:binomIneq}
    \sum_{k=0,\,k\,\mathrm{mod}\,2=0}^{2q} 
    \frac{\Gamma(q+1)c^{k/2}}{
    	\Gamma(\frac{k}{2}+1)
    	\Gamma(q-\frac{k}{2}+1)
    }
    \ge
    \sum_{k=1,\,k\,\mathrm{mod}\,2=1}^{2q-1} 
    \frac{\Gamma(q+1)c^{k/2}}{
    	\Gamma(\frac{k}{2}+1)
    	\Gamma(q-\frac{k}{2}+1)
    }
  \end{equation}
  and furthermore
  $$
    \sum_{k=1,\,k\,\mathrm{mod}\,2=1}^{2q-1} 
    \frac{\Gamma(q+1)c^{k/2}}{
    	\Gamma(\frac{k}{2}+1)
    	\Gamma(q-\frac{k}{2}+1)
    }
    = (1+c)^q + O(q^{-1/2}).
  $$
\end{lemma}
\begin{remark}
  In the proof of Lemma~\ref{lem:halfBinom}, we use a variety of
  transformations for hypergeometric functions, which can be
  found in the NIST Digital Library of Mathematical Functions
  \citep{NIST:DLMF} as well as in 
  Chapter 15 of \cite{ABRAMOWITZ_STEGUN}.  
  These include the following
  where $\abs{z}\le1$ and $a,b,c\in\complex$ such that 
  $\mathcal{R}(c-a-b)>0$.
  \begin{itemize}
  	\item Gauss' summation formula,
  	$
  	  \hgeo{2}{1}(a,b;c;1) =
  	  \frac{\Gamma(c)\Gamma(c-a-b)}{\Gamma(c-a)\Gamma(c-b)}
  	$
  	for $c\ne0,-1,-2,\ldots$
  	\citep[Eqn 15.1.20]{ABRAMOWITZ_STEGUN}.
  	\item Euler's integral transform,
  	$
  	\hgeo{2}{1}(a,b;c;z) =
  	\frac{\Gamma(c)}{\Gamma(b)\Gamma(c-b)}
  	\int_{0}^{1} t^{b-1}(1-t)^{c-b-1}(1-tz)^{-a}dt
  	$
  	for $\mathcal{R}(c) > \mathcal{R}(b) > 0$
  	\citep[Eqn 15.3.1]{ABRAMOWITZ_STEGUN}.
  	\item Pfaff's linear transform,
  	$
  	  \hgeo{2}{1}(a,b;c;z) =
  	  (1-z)^{-a}\hgeo{2}{1}(a,c-b;c;\frac{z}{z-1})
  	$
  	\citep[Eqn 15.3.4]{ABRAMOWITZ_STEGUN}.
  	\item Another linear transformation
  	$
  	  \hgeo{2}{1}(a,b;c;z) = 
  	  \frac{\Gamma(c)\Gamma(b-a)}{\Gamma(b)\Gamma(c-a)}
  	  (1-z)^{-a}\hgeo{2}{1}(
  	  a, c-b; a-b+1; \frac{1}{z-1}
  	  ) + 
  	  \frac{\Gamma(c)\Gamma(a-b)}{\Gamma(a)\Gamma(c-b)}
  	  (1-z)^{-b}\hgeo{2}{1}(
  	    b, c-a; b-a+1; \frac{1}{z-1}
  	  )
  	$
  	for $\abs{\mathrm{arg}(1-z)}<\pi$
  	\citep[Eqn 15.3.8]{ABRAMOWITZ_STEGUN}.
  \end{itemize}
  We also make use of Gautschi’s inequality for the 
  ratio of two Gamma functions, 
  $
  x^{1-s} < \frac{\Gamma(x+1)}{\Gamma(x+s)} <
  (x+1)^{1-s}
  $
  for $s\in(0,1)$
  \citep[Eqn. 5.6.4]{NIST:DLMF}.
\end{remark}
\begin{proof}
  We first note that the lefthand side of 
  Equation~\ref{eqn:binomIneq} is just $(1+c)^q$, which
  can be written as the generalized hypergeometric function
  ${}_1F_0(-q;;-c)$.
  
  For the righthand side, we rewrite it as 
  $$
    \mathrm{RHS}(\mathrm{Eqn}\,\ref{eqn:binomIneq}) = 
    \sum_{r=1}^{q} 
    \frac{\Gamma(q+1)c^{r-1/2}}{
    	\Gamma(r+\frac{1}{2})
    	\Gamma(q-r+\frac{3}{2})
    } = \frac{\Gamma(q+1)}{\sqrt{c}}\sum_{r=1}^{q} 
    \frac{c^{r}}{
      \Gamma(r+\frac{1}{2})
      \Gamma(q-r+\frac{3}{2})
    }
  $$
  and note that the ratio of consecutive terms is
  $$
    \left[{
        \frac{c^{r+1}}{
     	\Gamma(r+\frac{3}{2})
     	\Gamma(q-r+\frac{1}{2})
        }
   	}\right]\left[{
        \frac{
   		\Gamma(r+\frac{1}{2})
   		\Gamma(q-r+\frac{3}{2})
   	    }{c^{r}}
    }\right]
    = \left(\frac{q+\frac{1}{2}-r}{\frac{1}{2}+r}\right)c.
  $$
  After rescaling, we have the first $q$ terms of the 
  Gaussian hypergeometric 
  function $\hgeo{2}{1}(-q-\frac{1}{2},1;\frac{1}{2};-c)$ by noting 
  that
  \begin{align*}
    \hgeo{2}{1}\left(-q-\frac{1}{2},1;\frac{1}{2};-c\right) &=
    1 + \sum_{r=1}^{q} 
    \frac{\sqrt{\pi}\Gamma(q+\frac{3}{2})c^{r}}{
    	\Gamma(r+\frac{1}{2})
    	\Gamma(q-r+\frac{3}{2})
    } +
    \sum_{r>q} \frac{(-q-\frac{1}{2})_r}{(\frac{1}{2})_r}(-c)^r\\
    &=
    1 + \sum_{r=1}^{q} 
    \frac{\sqrt{\pi}\Gamma(q+\frac{3}{2})c^{r}}{
    	\Gamma(r+\frac{1}{2})
    	\Gamma(q-r+\frac{3}{2})
    } +
    c^{q+1}
    \hgeo{2}{1}\left(\frac{1}{2},1;q+\frac{1}{2};-c\right)
  \end{align*}
  where $(n)_r = n(n+1)\ldots(n+r-1)$ is the Pochhammer symbol
  or rising factorial.  Rearranging the above terms gives
  \begin{multline}
    \label{eqn:hgeoDiff}
    \frac{\Gamma(q+1)}{\sqrt{c}}
    \sum_{r=1}^{q} 
    \frac{c^{r}}{
    	\Gamma(r+\frac{1}{2})
    	\Gamma(q-r+\frac{3}{2})
    } \\= 
    \frac{\Gamma(q+1)}{\sqrt{c\pi}\Gamma(q+\frac{3}{2})}\left\{
      \hgeo{2}{1}\left(-q-\frac{1}{2},1;\frac{1}{2};-c\right) -
      c^{q+1}
      \hgeo{2}{1}\left(\frac{1}{2},1;q+\frac{1}{2};-c\right) -
      1
    \right\}
  \end{multline}
  
  For the first hypergeometric function in Equation~\ref{eqn:hgeoDiff},
  we apply a linear transformation formula 
  \citep[Eqn 15.3.8]{ABRAMOWITZ_STEGUN}, upper bounding the second
  term below by setting $1/(1+c)$ to $1$, and then 
  using Gauss' summation formula to get
  \begin{align*}
    \hgeo{2}{1}\left(
      -q-\frac{1}{2}, 1; \frac{1}{2}; -c
    \right)
    &= 
    (1+c)^{q-1/2}\frac{\Gamma(\frac{1}{2})\Gamma(q+\frac{3}{2})
    }{\Gamma(1)\Gamma(q+1)}
    \hgeo{2}{1}\left(
      -q-\frac{1}{2}, -\frac{1}{2}; -q-\frac{1}{2}; \frac{1}{1+c}
    \right)\\
    &~~+
    (1+c)^{-1}\frac{\Gamma(\frac{1}{2})\Gamma(-q-\frac{3}{2})
    }{\Gamma(-q-\frac{1}{2})\Gamma(-\frac{1}{2})}
    \hgeo{2}{1}\left(
    1, q+1; q+\frac{5}{2}; \frac{1}{1+c}
    \right)\\
    &\le (1+c)^{q-1/2}\frac{\sqrt{\pi}\Gamma(q+\frac{3}{2})
    }{\Gamma(q+1)}
    \hgeo{1}{0}\left( -\frac{1}{2}; \frac{1}{1+c}\right) +
    \frac{
    	\hgeo{2}{1}(1, q+1; q+\frac{5}{2}; 1)
    }{2(1+c)(q+\frac{3}{2})}\\  
    &= (1+c)^{q-1/2}\frac{\sqrt{\pi}\Gamma(q+\frac{3}{2})
    }{\Gamma(q+1)}
    \sqrt{1-\frac{1}{1+c}}
    +
    \frac{
      \frac{\Gamma(q+\frac{5}{2})\Gamma(\frac{1}{2})}{
      	\Gamma(q+\frac{3}{2})\Gamma(\frac{3}{2})
      }
    }{2(1+c)(q+\frac{3}{2})}\\
    &= (1+c)^{q}\frac{\sqrt{c\pi}\Gamma(q+\frac{3}{2})
    }{\Gamma(q+1)}
    +
    \frac{ 1 }{1+c}.
  \end{align*}
  
  For the second hypergeometric function in Equation~\ref{eqn:hgeoDiff},
  we apply the Pfaff transform and then Euler's integral transform to 
  get
  \begin{align*}
    \hgeo{2}{1}\left(\frac{1}{2},1;q+\frac{1}{2};-c\right)
    &= (1+c)^{-1/2}\hgeo{2}{1}\left(
      \frac{1}{2}, q-\frac{1}{2} ; q+\frac{1}{2};\frac{c}{c+1}
    \right)\\
    &= (1+c)^{-1/2} \frac{\Gamma(q+\frac{1}{2})}{\Gamma(q-\frac{1}{2})}
      \int_{0}^{1} t^{q-3/2}\left(1-\frac{ct}{c+1}\right)^{-1/2}dt\\
    &\le (1+c)^{-1/2} \frac{\Gamma(q+\frac{1}{2})}{\Gamma(q-\frac{1}{2})}
      \int_{0}^{1} t^{q-3/2}dt\left(1-\frac{c}{c+1}\right)^{-1/2}
    = 1.      
  \end{align*}

  Putting the above bounds into Equation~\ref{eqn:hgeoDiff} and
  upper bounding with Gautschi’s inequality
  we get the desired result:
  \begin{align*}
  \sum_{r=1}^{q} 
  \frac{\Gamma(q+1)c^{r-1/2}}{
  	\Gamma(r+\frac{1}{2})
  	\Gamma(q-r+\frac{3}{2})
  } &\le 
  \frac{\Gamma(q+1)}{\sqrt{c\pi}\Gamma(q+\frac{3}{2})}\left\{
  (1+c)^{q}\frac{\sqrt{c\pi}\Gamma(q+\frac{3}{2})
  }{\Gamma(q+1)}
  +
  \frac{ 1 }{1+c} -
  c^{q+1} - 1
  \right\}\\
  &= (1+c)^q - 
  \frac{\Gamma(q+1)}{\sqrt{c\pi}\Gamma(q+\frac{3}{2})}
  \left\{
    \frac{c}{1+c} + c^{q+1}
  \right\}\\
  &\le (1+c)^q - \frac{\sqrt{q+1}}{q+\frac{1}{2}}
  \frac{1}{\sqrt{\pi}}
  \left\{
  \frac{\sqrt{c}}{1+c} + c^{q+1/2}
  \right\}\\
  &\le (1+c)^q - 
  \frac{3}{\sqrt{2\pi}}
  \sqrt{
  	\frac{1}{q+\frac{1}{2}} + \frac{1/2}{(q+\frac{1}{2})^2}
  }.
  \end{align*}
\end{proof}

\begin{lemma}
	\label{lem:multiSum}
	For $p,n\ge1$, let $c_1,\ldots,c_n\in\real^+$.  Then,
	$$
	\sum_{k_1+\ldots+k_n=2p}
	\frac{\Gamma(p+1)}{\Gamma(k_1/2+1)\ldots\Gamma(k_n/2+1)}
	\prod_{i=1}^nc_i^{k_i/2}
	\le
	2^{n-1}
	(c_1+\ldots+c_n)^p
	$$
	where the sum is taken over all integer 
	compositions of $2p$.
\end{lemma}
\begin{proof}
	We let $\Delta_{2p}^n$ denote the discrete simplex
	$$
	\Delta_{2p}^n = \left\{
	(k_1,\ldots,k_n)\in\natural^n\,|\,
	k_i\ge0\,\forall i\text{ and }
	k_1+\ldots+k_n=2p
	\right\}
	$$
	being the set of all $n$-long integer compositions of 
	$2p\in\natural$.  As $2p$ is an even integer, we can 
	consider even compositions $(k_1,\ldots,k_n)$ such that
	$\sum_{i=1}^n k_i=2p$ and $k_i\text{ mod }2=0$ for all 
	$i=1,\ldots,n$, and we denote $k_i = 2l_i$.
	From the multinomial theorem, summing over 
	only even compositions of $2p$ gives
	$$
	\sum_{2l_1+\ldots+2l_n=2p}
	\frac{\Gamma(p+1)}{\Gamma(l_1+1)\ldots\Gamma(l_n+1)}
	\prod_{i=1}^nc_i^{l_i}
	=
	(c_1+\ldots+c_n)^p.
	$$
	The collection of even compositions $(2l_1,\ldots,2l_n)$
	forms a discrete subsimplex of $\Delta_{2p}^n$ isomorphic 
	to $\Delta_p^n$. We further decompose the remaining 
	not-strictly-even  (NSE) compositions of $2p$ into 
	$\sum_{m=1}^{n/2}{n\choose2m} = 2^{n-1}-1$, 
	disjoint subsimplices by indicating which entries in 
	the composition are even and which are odd.	This
	summation
	follows directly from the identity 
	$\sum_{m=0}^n (-1)^m{n\choose m}=0$.
	Because $2p$ is even, none of the NSE subsplicies can 
	have cardinality more than 
	$\abs{\Delta_{p}^n} = {p+n-1\choose n-1}$.
	We will denote an NSE subsimplex as 
	$\tilde{\Delta}_{2p}^n(o)$ where $o = 0,2,4,\ldots$ is
	the number of odd entries in the composition. 
	Each subsimplex with exactly $o$ odd entries
	is isomorphic to the others via translation. Hence,
	without loss of generality, we choose the simplex 
	$\tilde{\Delta}_{2p}^n(o)$ to be the one with odd
	entries $k_1,\ldots,k_o$ and even entries 
	$k_{o+1},\ldots,k_n$.
	
	Noting that the subsimplex of even compositions is 
	$\tilde{\Delta}_{2p}^n(0)$,
	we first prove that
	\begin{multline*}
		(c_1+\ldots+c_n)^p = 
		\sum_{k\in\tilde{\Delta}_{2p}^n(0)}
		\frac{\Gamma(p+1)}{\Gamma(k_1/2+1)\ldots\Gamma(k_n/2+1)}
		\prod_{i=1}^nc_i^{k_i/2}\\
		\ge
		\sum_{k\in\tilde{\Delta}_{2p}^n(2)}
		\frac{\Gamma(p+1)}{\Gamma(k_1/2+1)\ldots\Gamma(k_n/2+1)}
		\prod_{i=1}^nc_i^{k_i/2}.
	\end{multline*}
	by summing along a ``row'' of $\Delta_{2p}^n$.
	We recall that entries $k_1$ and $k_2$
	are odd in $\tilde{\Delta}_{2p}^n(2)$.  
	By fixing 
	the remaining $k_3,\ldots,k_n$, denoting 
	$2q = 2p-k_3-\ldots-k_n$, and applying 
	Lemma~\ref{lem:halfBinom}, we note that 
	\begin{multline*}
		\frac{\Gamma(p+1)}{\Gamma(k_3/2+1)\ldots\Gamma(k_n/2+1)}
		\prod_{i=3}^nc_i^{k_i/2}
		\sum_{k_1+k_2=2q,\text{ even}}
		\frac{
			c_1^{k_1/2}c_2^{k_2/2}
		}{
			\Gamma(k_1/2+1)\Gamma(k_2/2+1)
		}\\
		\ge
		\frac{\Gamma(p+1)}{\Gamma(k_3/2+1)\ldots\Gamma(k_n/2+1)}
		\prod_{i=3}^nc_i^{k_i/2}
		\sum_{k_1+k_2=2q,\text{ odd}}
		\frac{
			c_1^{k_1/2}c_2^{k_2/2}
		}{
			\Gamma(k_1/2+1)\Gamma(k_2/2+1)
		}.
	\end{multline*}
	Applying this for every choice of $k_3,\ldots,k_n$
	demonstrates that the multinomial sum over all elements
	in $\tilde{\Delta}_{2p}^n(0)$ is greater or equal to
	the sum over all elements in any of the 
	$\tilde{\Delta}_{2p}^n(2)$.
	
	Repeating this argument shows that the multinomial
	sum over $\tilde{\Delta}_{2p}^n(o)$ is greater than
	or equal to the sum over $\tilde{\Delta}_{2p}^n(o+2)$.
	Denoting 
	$\zeta = \{ 
	\tilde{\Delta}_{2p}^n(2l)\,:\, l=1,\ldots,\lfloor n/2\rfloor
	\}$ to be the set of all subsimplices of $\Delta_{2p}^n$,
	we conclude that 
	\begin{multline*}
		\sum_{k_1+\ldots+k_n=2p}
		\frac{\Gamma(p+1)}{\Gamma(k_1/2+1)\ldots\Gamma(k_n/2+1)}
		\prod_{i=1}^nc_i^{k_i/2}\\
		= 
		\sum_{\Delta\in\zeta}\sum_{ {\bf k} \in \Delta}
		\frac{\Gamma(p+1)}{\Gamma(k_1/2+1)\ldots\Gamma(k_n/2+1)}
		\prod_{i=1}^nc_i^{k_i/2}\\
		\le
		\abs{\zeta}
		(c_1+\ldots+c_n)^p
		\le
		2^{n-1}
		(c_1+\ldots+c_n)^p
	\end{multline*}
	by noting that 
	$
	\abs{\zeta} = \sum_{l=0}^{\lfloor n/2\rfloor}{n\choose 2l}
	= 2^{n-1}.
	$
\end{proof}

\begin{proof}[Proof of Theorem~\ref{thm:globalGamma}]
  Beginning from the proof of Theorem~\ref{thm:localGamma}, we
  recall that we can  
  apply an affine transformation to write
  $
    \gamma_i(\pi_i) = \frac{m_i(n-m_i-1)}{n-1}\sum_{j=1}^n
    \delta_{i,j}\lmb_{i,\pi_i(j)}+m_i\bar{\lmb}_{-i}.
  $
  Thus, our global test statistic can be written as 
  $$
    \gamma(\boldsymbol{\pi}) = \sum_{i=1}^n\left[
    \frac{m_i(n-m_i-1)}{n-1}\sum_{j=1}^n
    \delta_{i,j}\lmb_{i,\pi_i(j)}\right] + \sum_{i=1}^n m_i\bar{\lmb}_{-i}.
  $$
  Inference based on the permutation test will not be affected 
  by the constant shift term $\sum_{i=1}^n m_i\bar{\lmb}_{-i}$.
  Hence, we can proceed by considering 
  $
    T(\boldsymbol{\pi}) = \sum_{i=1}^n \eta_i
    T_i(\pi_i)
  $
  for $\eta = [m_i(n-m_i-1)/(n-1)]$ and $T_i=\sum_{j=1}^n
  \delta_{i,j}\lmb_{i,\pi_i(j)}$.
  
  We recall that $\boldsymbol{\pi}=(\pi_1,\ldots,\pi_n)$ is
  such that $\pi_i$ and $\pi_j$ are independent random permutations 
  for $i\ne j$. Then, we  bound
  the $p$th moment of $T(\boldsymbol{\pi})$
  as follows:
  \begin{align*}
  	\xv T(\boldsymbol{\pi})^p
  	&= \sum_{k_1+\ldots+k_n=p}{p\choose k_1,\ldots,k_n}
  	\prod_{i=1}^n \eta_i^{k_i} \xv \left[T_i(\pi_i)^{k_i}\right]\\
  	&\le \sum_{k_1+\ldots+k_n=p}{p\choose k_1,\ldots,k_n}
  	\prod_{i=1}^n \eta_i^{k_i} 
  	\frac{k_i!(n-1)^{k_i}s_i^{k_i}}{2^{k_i/2}m_i^{3k_i/2}
  		\Gamma(\frac{k_i}{2}+1)}  \\
  	&= \frac{(n-1)^p}{2^{p/2}}
  	\sum_{k_1+\ldots+k_n=p}
  	\frac{\Gamma(p+1)}{\Gamma(\frac{k_1}{2}+1)\ldots\Gamma(\frac{k_n}{2}+1)}
  	\prod_{i=1}^n  
  	\left(\frac{\eta_i^{2}s_i^{2}}{m_i^{3}}\right)^{k_i/2}  \\
  	&= \left(\frac{n-1}{2^{1/2}}\right)^{p}
  	\frac{\Gamma(p+1)}{\Gamma(\frac{p}{2}+1)}
  	\sum_{k_1+\ldots+k_n=p}
  	\frac{
  	  \Gamma(\frac{p}{2}+1)
  	}{
  	  \Gamma(\frac{k_1}{2}+1)\ldots\Gamma(\frac{k_n}{2}+1)
    }
  	\prod_{i=1}^n  
  	\left(\frac{\eta_i^{2}s_i^{2}}{m_i^{3}}\right)^{k_i/2}  \\
  	&\le \left(\frac{n-1}{2^{1/2}}\right)^{p}
  	\frac{\Gamma(p+1)}{\Gamma(\frac{p}{2}+1)}
  	2^{n-1}\left(
  	  \sum_{i=1}^n\frac{\eta_i^{2}s_i^{2}}{m_i^{3}}
  	\right)^{p/2}
  \end{align*}
  where the first inequality comes from 
  Theorem A.4 of \citep{KASHLAK_KHINTCHINE2020} and
  the second inequality comes from Lemma~\ref{lem:multiSum}
  above.
  Symmetrizing the statistic $T$ with $\boldsymbol{\pi}'$
  is an iid copy of $\boldsymbol{\pi}$ and 
  application of Markov/Chernoff's inequality gives 
  \begin{align*}
    \prob{\abs{T(\boldsymbol{\pi})}>t}  
    &\le
    \inf_{\lmb>0}\ee^{-\lmb t}
    \xv\ee^{\lmb(T(\boldsymbol{\pi})-T(\boldsymbol{\pi}'))}
    \\
    &\le
    \inf_{\lmb>0}\ee^{-\lmb t}\left[
      1 + \sum_{p=1}^\infty \frac{\lmb^p}{p!}\xv
      (T(\boldsymbol{\pi})-T(\boldsymbol{\pi}'))^p
    \right]\\ 
    &\le
    \inf_{\lmb>0}\ee^{-\lmb t}\left[
    1 + \sum_{p=1}^\infty \frac{\lmb^{2p}2^{2p}}{{(2p)}!}\xv
    T(\boldsymbol{\pi})^{2p}
    \right]\\ 
    &\le
    \inf_{\lmb>0}\ee^{-\lmb t}\left[
    1 + \sum_{p=1}^\infty \frac{\lmb^{2p}2^{2p}}{{(2p)}!}
    \left(\frac{n-1}{2^{1/2}}\right)^{2p}
    \frac{(2p)!}{{p}!}
    2^{n-1}\left(
    \sum_{i=1}^n\frac{\eta_i^{2}s_i^{2}}{m_i^{3}}
    \right)^{p}
    \right]
    \\ 
    &\le
    \inf_{\lmb>0}\ee^{-\lmb t}\left[
    1 + \sum_{p=1}^\infty \frac{2^{n-1}}{p!}\left\{
      {2(n-1)^2\lmb^2}
    \right\}^p\left(
    \sum_{i=1}^n\frac{\eta_i^{2}s_i^{2}}{m_i^{3}}
    \right)^{p}
    \right] 
    \\ 
    &\le
    \inf_{\lmb>0}\ee^{-\lmb t}\left[
    1 + \sum_{p=1}^\infty \frac{(2^n\lmb^2)^p}{p!}
    \left(
    \sum_{i=1}^n\frac{(n-m_i-1)^2s_i^{2}}{m_i}
    \right)^{p}
    \right] 
    \\
    &\le
    \inf_{\lmb>0}\ee^{-\lmb t}\exp\left\{
      \left(
      {2^n\lmb^2}
      \right)
      \sum_{i=1}^n\frac{(n-m_i-1)^2s_i^{2}}{m_i}
    \right\}
    \\
    &\le
    \exp\left\{
      -\frac{t^2}{2^{n+2}}\left(
      \sum_{i=1}^n\frac{(n-m_i-1)^2s_i^{2}}{m_i}
      \right)^{-1}
    \right\}
    = \exp\left\{ -\frac{t^2}{2^{n+2}\varpi^2} \right\}
  \end{align*}
  where $\varpi^2=\sum_{i=1}^n\frac{(n-m_i-1)^2s_i^{2}}{m_i}$
  acts like a variance for this sub-Gaussian bound.
  
  Lastly, we modify the proof of Proposition~2.5 of 
  \cite{KASHLAK_KHINTCHINE2020} to improve that bound.
  We first note that 
  $\var{T_i(\pi)} = 
  s_i^2(\frac{1}{n-m_i-1} + \frac{1}{m_i})$.
  Recalling the defining of $\eta_i$, 
  we get that 
  \begin{multline*}
    \var{T(\pi)} = \sum_{i=1}^n s_i^2\eta_i^2\left(
      \frac{1}{n-m_i-1} + \frac{1}{m_i}
    \right)\\
    = \sum_{i=1}^n s_i^2\left(
      \frac{m_i(n-m_i-1)^2 + m_i^2(n-m_i-1)}{(n-1)^2} 
    \right)
    = \sum_{i=1}^n \eta_is_i^2 := \upsilon^2.
  \end{multline*}
  Thus, we apply the proof of Proposition~2.5 of
  \cite{KASHLAK_KHINTCHINE2020}
  to get
  \begin{equation}
    \label{eqn:incBeta}
    \prob{
      \ee^{ -\frac{T(\pi)^2}{2^{n+2}\varpi^2} } < u
    } \le C_0 I\left(
      u; 2^{n}\frac{\varpi^2}{\upsilon^2}
    ,\frac{1}{2}\right)
  \end{equation}
  where $I$ is the regularized incomplete beta function and
  $$
    C_0 = \frac{
      \left(
        2^{n}\frac{\varpi^2}{\upsilon^2}
      \right)^{1/2}
      \Gamma\left(2^{n}\frac{\varpi^2}{\upsilon^2}\right)
    }{
      \Gamma\left(\frac{1}{2} + 2^{n}\frac{\varpi^2}{\upsilon^2}\right)
    } \approx 1.
  $$
  However, the presence of $2^n$ in both the exponent and
  the beta parameter in Equation~\ref{eqn:incBeta} makes 
  this numerically impossible to compute for moderate to
  large sample sizes $n$.  Thus, we apply the asymptotic 
  formula for the incomplete beta function detailed in 
  \cite{DOMAN1996} to get a numerically stable 
  equation for the tail probability.
  
  In \cite{DOMAN1996}, 
  \begin{equation}
    \label{eqn:doman}
    I(x;a,b) \sim Q\left(-g \log x;b\right)
    + \frac{\Gamma(a+b)}{\Gamma(a)\Gamma(b)}x^{g}
    \sum_{k=0}^\infty T_k(b,x)/g^{k+1}
  \end{equation}
  where $Q$ is the upper regularized incomplete gamma function,
  $g = a + (b-1)/2$ and the $T_k$ are power series related to 
  the $\sinh()$ function and the Bernoulli polynomials.
  In our context, $g = 2^{n}\varpi^2/\upsilon^2 - 1/4\approx
  2^{n}\varpi^2/\upsilon^2$.  
  The first piece of Equation~\ref{eqn:doman}
  becomes
  $$
    Q(-g\log x;b) = Q\left(
      \left\{2^{n}\frac{\varpi^2}{\upsilon^2}-\frac{1}{4}\right\}
      \frac{T^2}{2^{n+2}\varpi^2}
      ;\frac{1}{2}
    \right) = 
    Q\left(
      \frac{T^2}{4\upsilon^2} - \frac{T^2}{2^{n+4}\varpi^2}
    ;\frac{1}{2}
    \right).
  $$
  The second term of the asymptotic expansion contains a few
  subparts.  First, we can use Stirling's approximation
  to show that for $a\rightarrow\infty$ with $b$ fixed
  \begin{align*}
    \frac{\Gamma(a+b)}{\Gamma(a)\Gamma(b)} 
    &\sim
    \frac{a^b}{\Gamma(b)} = 
    \frac{2^{n/2}\varpi}{\upsilon\sqrt{\pi}}
  \end{align*}
  We also have that $x^g = \exp( -T^2/2\upsilon^2 )$.
  For the final series term, we only consider the 
  first term (n=0) as $g^{k+1}\sim 2^{nk}$ making 
  the remainder negligible.  The result is
  $$
    \sum_{k=0}^\infty \frac{T_k(b,x)}{g^{k+1}}
    \sim 
    \frac{
      T_0(1/2,\ee^{-T^2/2^{n+2}\varpi^2})
    }{
      2^{n}\varpi^2/\upsilon^2
    }
    \sim 
    \frac{
    	\abs{T}^3/2^{3n/2}\varpi^3
    }{
    	2^{n}\varpi^2/\upsilon^2
    }
    = \frac{
      \upsilon^2\abs{T}^3
    }{
      2^{5n/2}\varpi^5
    }
  $$
  Combining all three pieces gives the expression
  $$
    \frac{2^{n/2}\varpi}{\upsilon\sqrt{\pi}}
    \ee^{-T^2/2\upsilon^2}
    \frac{
    	\upsilon^2\abs{T}^3
    }{
    	2^{5n/2}\varpi^5
    }
    = \frac{1}{\sqrt{\pi}}
    \frac{
    	\upsilon\abs{T}^3
    }{
    	2^{2n}\varpi^4
    }
    \ee^{-T^2/2\upsilon^2}
  $$
  Combining with the first part of the asymptotic expansion,
  we conclude that 
  $$
    \prob{
    	\ee^{ -\frac{T(\pi)^2}{2^n\varpi^2} } < u
    }
    \sim
    Q\left(
    \frac{T^2}{4\upsilon^2}
    ;\frac{1}{2}
    \right)
    +
   \frac{1}{\sqrt{\pi}}
   \frac{
   	\upsilon\abs{T}^3
   }{
   	2^{2n}\varpi^4
   }
   \ee^{-T^2/2\upsilon^2}
  $$
  and finally that
  $$
  \prob{ 
  	\abs*{\gamma(\boldsymbol{\pi})- \sum_{i=1}^n m_i\bar{\lmb}_{-i}} \ge \gamma  }
  \sim
  Q\left(
  \frac{\gamma^2}{4\upsilon^2}
  ;\frac{1}{2}
  \right)
  +
  O(2^{-2n})
  $$
\end{proof}

\section{Comparison with the Gaussian Approximation}

Figure~\ref{fig:pvCompGauss} further reprises the analysis
displayed in Figure~\ref{fig:pvComp} but compares the 
computation-based permutation test to the Gaussian approximation.
The Gaussian approximation is not valid for the Alberta electoral
dataset as can be seen by the wild disagreement between these
two methods.  For Moran's statistic with 1-NN and 2-NN weight matrices,
the Gaussian approach appears prone to overstating the significance
of the local autocorrelation whereas for the 3-NN weight matrix, 
it understates the significance of many ridings.
For Geary's statistic, there is little agreement between the 
permutation p-values and the Gaussian p-values.  Albeit, this 
departure has already been noted in \cite{ANSELIN2019,SEYAS2020}
who recommend the permutation test for Geary's statistic.

\begin{figure}
	\begin{center}
		\includegraphics[width=0.45\textwidth]{\PICDIR/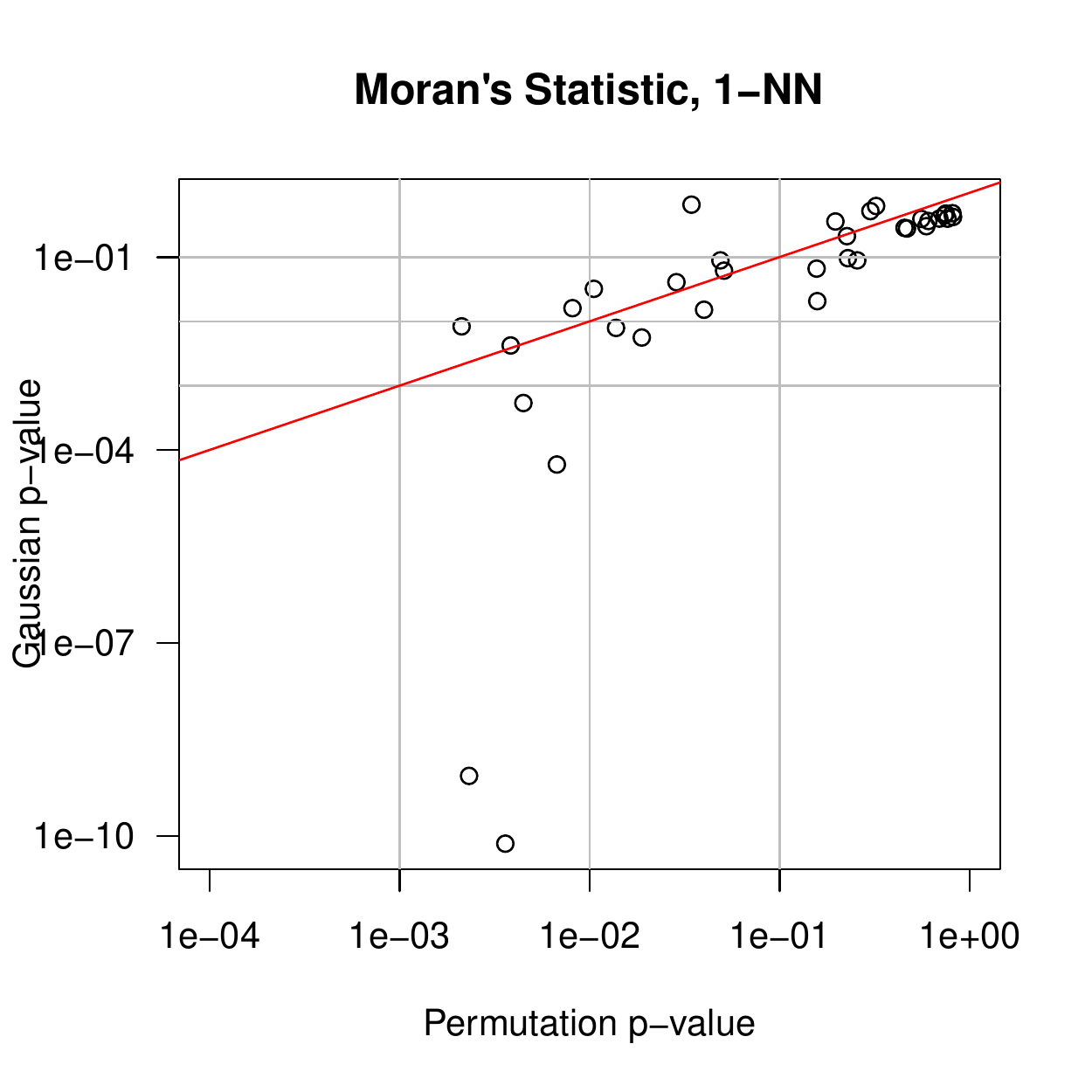}
		\includegraphics[width=0.45\textwidth]{\PICDIR/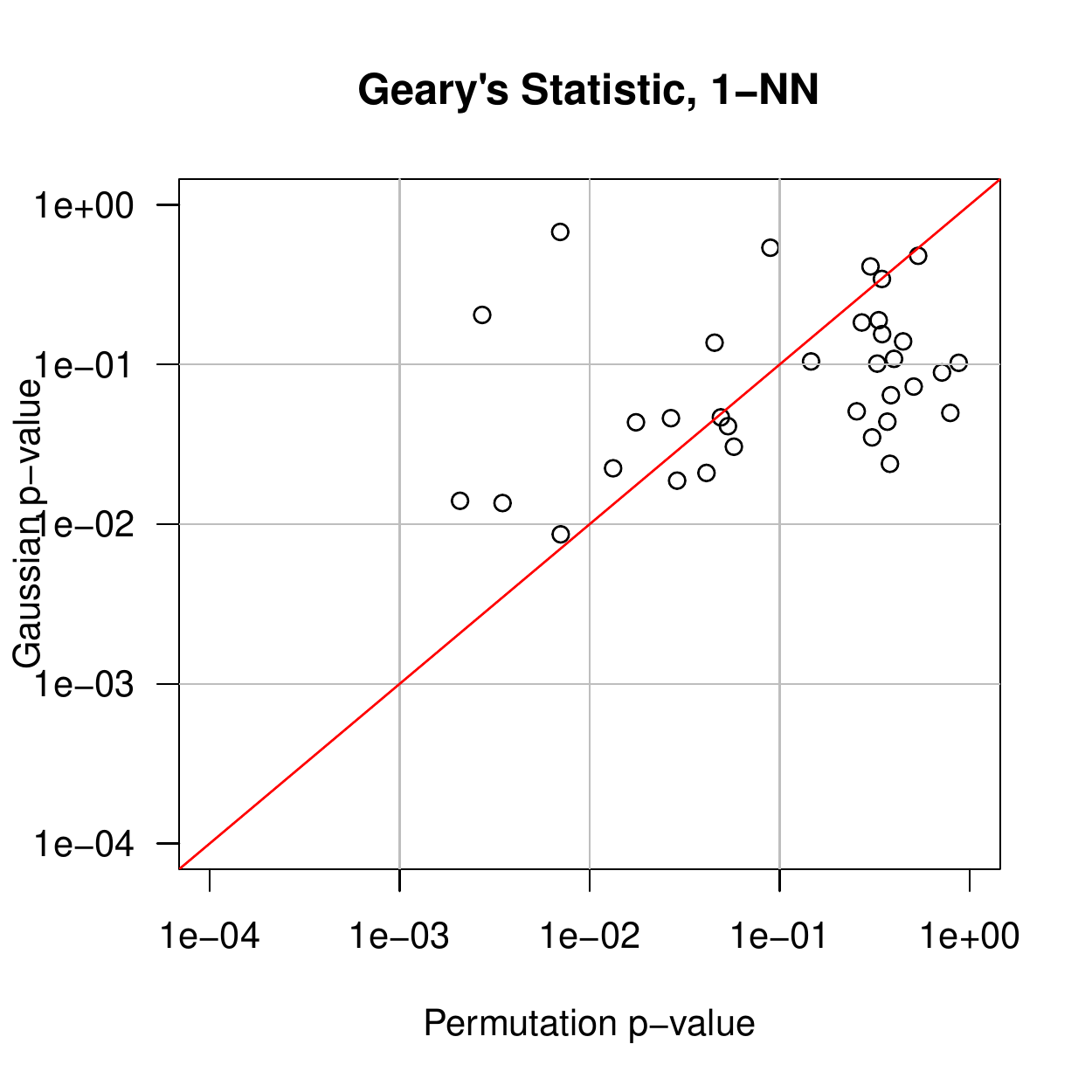}
		\includegraphics[width=0.45\textwidth]{\PICDIR/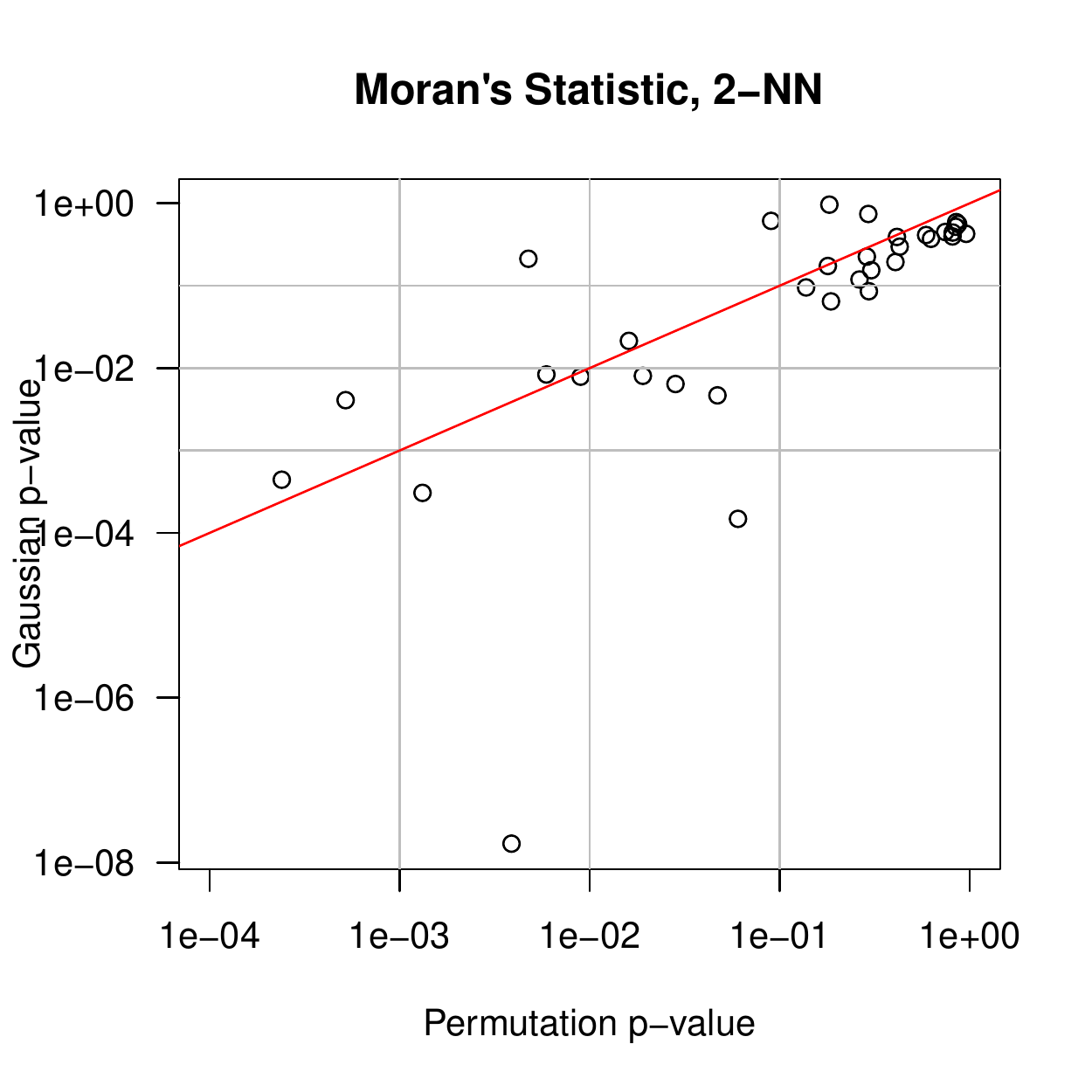}
		\includegraphics[width=0.45\textwidth]{\PICDIR/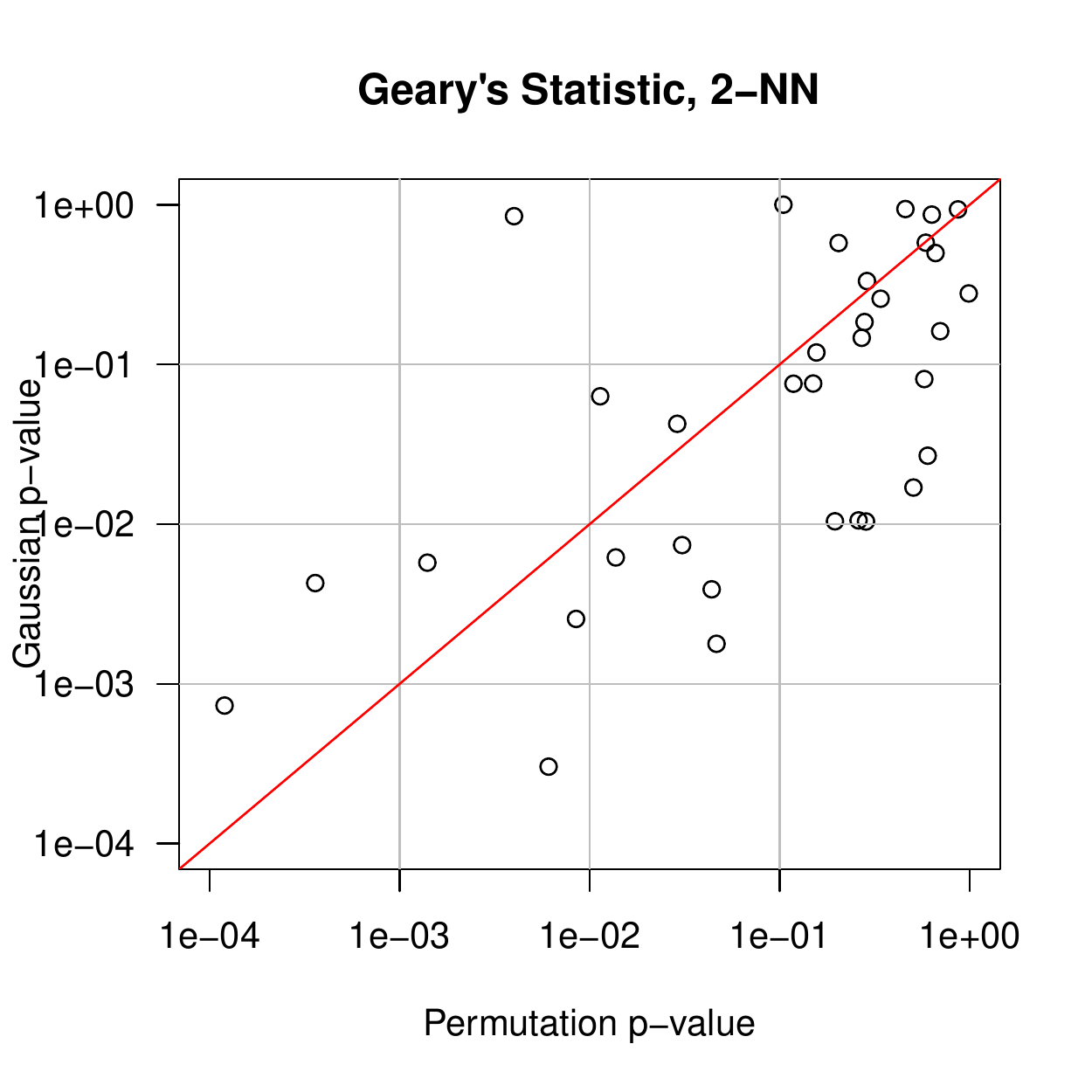}
		\includegraphics[width=0.45\textwidth]{\PICDIR/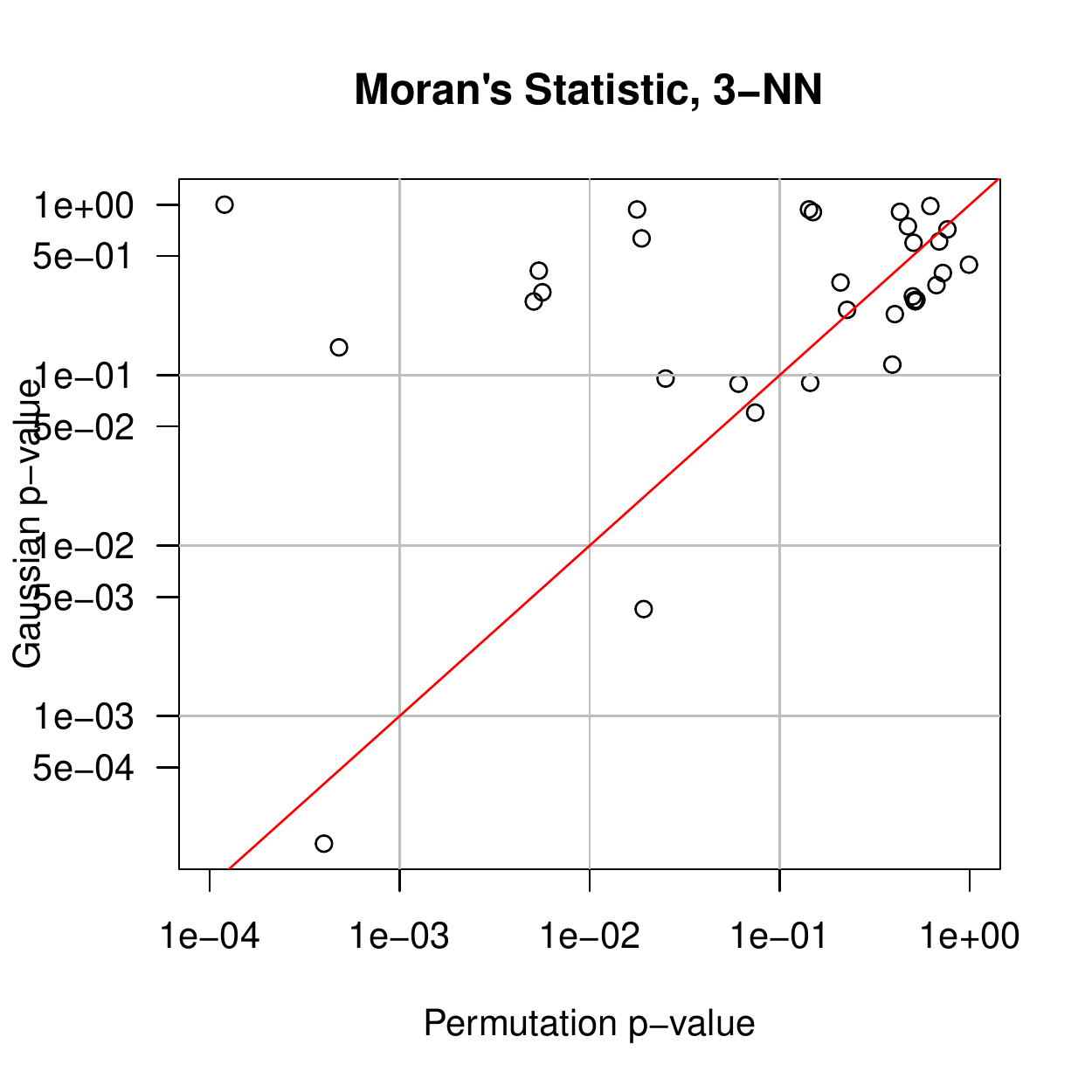}
		\includegraphics[width=0.45\textwidth]{\PICDIR/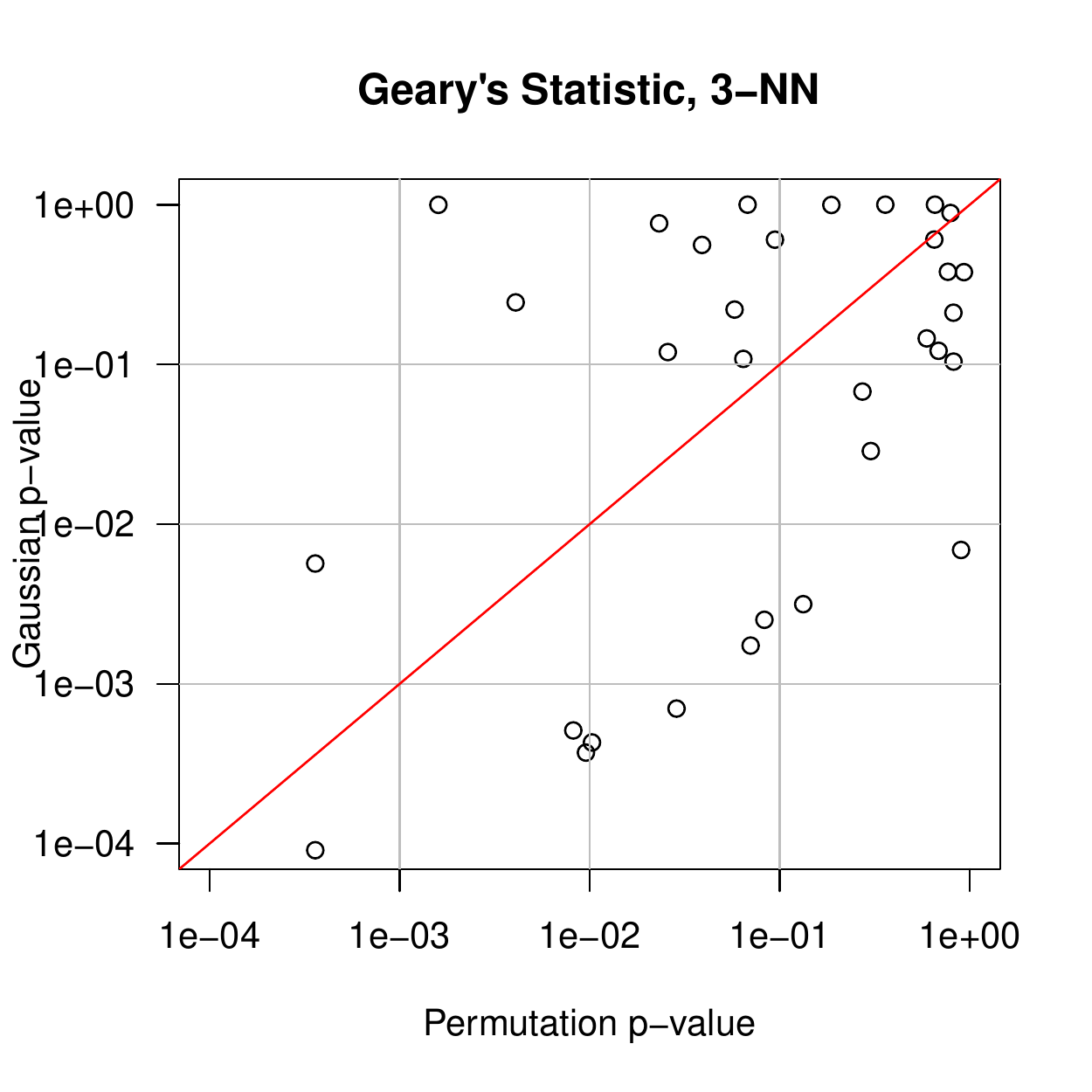}
	\end{center}
	\caption{
		\label{fig:pvCompGauss}
		A comparison of the p-values produced by a simulation-based 
		permutation test with 50,000 permutations per vertex
		and the p-values produced by assuming the test statistic
		follows a Gaussian distribution.  The left column considers
		Moran's statistic; the right column considers Geary's
		statistic.  The three rows from top to bottom consider
		the 1, 2, and 3-nearest neighbours weight matrix, respectively.
	}
\end{figure}

\end{document}